\documentclass[10pt,journal]{IEEEtran}
\usepackage{cite}
\usepackage{float}
\usepackage{graphicx}
\usepackage{amsmath}
\usepackage{times}
\usepackage{graphicx}
\usepackage{latexsym}
\usepackage{graphicx}
\usepackage{bm}
\usepackage{amssymb}
\usepackage[center]{caption2}
\usepackage{stfloats}
\usepackage{cases}
\usepackage{array}
\usepackage{setspace}
\usepackage{fancyhdr}
\usepackage{color}
\usepackage{subfigure}
\usepackage{hyperref}
\usepackage{multirow}
\usepackage{amsthm}
\usepackage{cases}
\usepackage{algpseudocode}
\usepackage[linesnumbered,ruled,vlined]{algorithm2e}
\usepackage[long]{optidef}
\usepackage{epstopdf}
\usepackage{verbatim}
\usepackage{xcolor}
\usepackage{enumitem}
\allowdisplaybreaks[4]

\theoremstyle{definition}
\newtheorem{theorem}{Theorem}
\newtheorem{lemma}{Lemma}

\newtheorem{proposition}{Proposition}

\begin{document}

\title{Rethinking the fundamental performance limits of integrated sensing and communication systems}
\author{\emph{Author}}
\author{\IEEEauthorblockN{Zhouyuan Yu, Xiaoling Hu, {\em Member, IEEE},  Chenxi Liu, {\em Senior Member, IEEE}, and Mugen Peng, {\em Fellow, IEEE}
\thanks{Z. Yu, X. Hu, C. Liu, and M. Peng  are with the State Key Laboratory of Networking and Switching Technology, Beijing University of Posts and
Telecommunications, Beijing 100876, China (e-mail: \{zhouyuanyu, xiaolinghu, chenxi.liu, pmg\}@bupt.edu.cn).}}}

\maketitle\thispagestyle{headings}
\setcounter{page}{1}
\begin{abstract}
Integrated sensing and communication (ISAC) has been recognized as a key enabler and feature of future wireless networks. In the existing works analyzing the performances of ISAC, discrete-time systems were commonly assumed, which, however, overlooked the impacts of temporal, spectral, and spatial properties. To address this issue, we establish a unified information model for the band-limited continuous-time ISAC systems. In the established information model, we employ a novel sensing performance metric, called the sensing mutual information (SMI). Through analysis, we show how the SMI can be utilized as a bridge between the mutual information domain and the mean squared error (MSE) domain. In addition, we illustrate the communication mutual information (CMI)-SMI and CMI-MSE regions to identify the performance bounds of ISAC systems in practical settings and reveal the trade-off between communication and sensing performances. Moreover, via analysis and numerical results, we provide two valuable insights into the design of novel ISAC-enabled systems: i) communication prefers the waveforms of random amplitude, sensing prefers the waveforms of constant amplitude, both communication and sensing favor the waveforms of low correlations with random phases; ii) There exists a linear positive proportional relationship between the allocated time-frequency resource and the achieved communication rate/sensing MSE.

\end{abstract}


\section{Introduction}
\par Next-generation networks are anticipated to provide massive wireless connectivity and high-precision wireless sensing capability for supporting numerous emerging applications such as smart factory industrial, extended reality, and vehicular networks\cite{LiuIntegrated2022, ZhangOverview}. For fulfilling this requirement, integrated sensing and communication (ISAC) is envisioned as a pivotal enabler where communication and sensing functionalities are co-designed to share hardware platform, frequency band, as well as signal processing modules, thus providing unprecedented 
synergy and integration gain\cite{LiuSeventy2023, ChiriyathRadar, FengJoint2020}. Owing to the immense potential of ISAC, its fundamental limits, which are of profound importance in guiding the design and theoretical analysis of practical ISAC systems, have recently sparked a surge in research attention and endeavors\cite{LuIntegrated}.

\par Generally speaking, communication is to recover data symbols embedded in transmitted signals from received signals, whose performance is typically evaluated by information-theoretic metrics including channel capacity, spectral/energy efficiency, and outage probability. By contrast, sensing is to extract information of interest about sensed objects from collected echo signals, whose performance is assessed through classical estimation-theoretic metrics including detection probability, false-alarm probability, mean squared error (MSE), as well as Cramér-Rao Bound (CRB)\cite{LiuSurvey, OuyangIntegrated}. By leveraging these well-defined basic metrics, some research efforts have been conducted to investigate the communication and sensing performance limits in ISAC \cite{ChalisePerformance2017, ChalisePerformance2018, XiaoWaveform2022, XiongFlowing, HuaGLOBECOM, RenFundamental, LiuCram, HuaMIMO, AnFundamental2023}. For instance, Chalise  \emph{et al}. analyzed the trade-off between the detection probability and the communication rate for a joint communication and passive radar system in \cite{ChalisePerformance2017}, and subsequently extended the analysis to a multi-static passive radar-communication system in \cite{ChalisePerformance2018}. In \cite{XiaoWaveform2022}, the authors derived the detection probability and ambiguity function for sensing as well as the symbol error rate and spectrum efficiency for communication, which were utilized to demonstrate the outperformance of the proposed full-duplex waveform. For characterizing the trade-off between target estimation and communication, the work \cite{XiongFlowing} obtained the CRB-minimization and rate-maximization points on the CRB-rate region in a point-to-point MIMO ISAC system. Later on, \cite{HuaGLOBECOM} and \cite{RenFundamental} further advanced the study to the extended target scenario and multicast multi-target scenario, respectively, and characterized the Pareto-optimal boundary of CRB-rate region. 

\par In addition to investigating through well-defined basic metrics, recent works have also been dedicated to connecting native communication and sensing metrics as well as constructing unified theoretical frameworks aimed at unveiling more insights into the fundamental limits of ISAC systems \cite{DongningMutual, TangMIMO2010, AhmadipourInformation, AhmadipourJoint, AhmadipourSystems}. From the information-theoretical perspective, the early seminal research \cite{DongningMutual} connected the communication metric of mutual information (MI) and the sensing metric of minimum MSE (MMSE) by the proposed well-known I-MMSE equation, which states that the derivative of MI with respect to the signal-to-noise ratio (SNR) equals half of the MMSE. This equation reveals that communication and sensing exhibit consistency in SNR, but conflict in determinism-randomness. Ahmadipour \emph{et al}. established a unified capacity-distortion performance metric for state-dependent memoryless channels with generalized feedback. Investigations showed that the capacity-distortion trade-off arises from a common choice of the waveform rather than other properties of the utilized codes \cite{AhmadipourInformation, AhmadipourJoint}. Besides, some new information metrics for sensing were defined to facilitate the trade-off with conventional communication metrics, e.g., the radar capacity that characterizes the number of identifiable targets\cite{GuerciJoint}, the radar estimation rate that describes the cancellation of the target parameters uncertainty per second \cite{ChiriyathInner, BlissCooperative}, as well as the sensing MI (SMI) that is defined as the MI between the received echo signals and target parameters \cite{YangMIMOradar, LiuDeterministic-Random}. Similarly, from the estimation-theoretical perspective, the authors in \cite{KumariPerformance} bridged the communication rate with the estimation-theoretic metric and proposed the communication MSE, which can be used to measure the average MSE of the ISAC system. In \cite{YuNon-Orthogonal}, the authors derived the communication CRB for an intelligent reflecting surface (IRS) assisted ISAC system, and revealed the inherent connection between the proposed communication CRB and traditional communication MI (CMI).

\par Nevertheless, the prevailing researches on the fundamental limits of ISAC consider only the discrete-time systems, in which the discrete Gaussian channel model is employed and the theoretical performance is portrayed using various communication-sensing metric pairs. However, this approach fails to capture the temporal, spectral, and spatial characteristics of practical band-limited continuous-time systems, leaving the impact of time-frequency-spatial resources on the performance bound and communication-sensing trade-off far from being well investigated. Motivated by the above, in this paper, we extend the discrete-time ISAC system to a more practical band-limited continuous-time ISAC system, and establish a unified ISAC information model that incorporates time, frequency, and spatial domain features. The main contributions of this work are summarized as follows:
\begin{itemize}
        \item Combining the information theory and the Nyquist sampling theorem, we develop an information model for band-limited discrete-time ISAC systems, which captures the features of both time, frequency, and spatial domains. Under such a framework, we derive the SMI for characterizing the sensing performance from the information-theoretical perspective, and reveal the inner connection between the SMI and the traditional estimation-theoretic metric MSE. Based on the proposed metrics, we develop the CMI-SMI and CMI-MSE regions to investigate the communication-sensing performance boundary and fundamental trade-off, thereby providing useful guidelines and insights for the time-frequency resource allocation and the waveform design in practical band-limited discrete-time ISAC systems.

        \item Through analytical and numerical investigation of the ISAC waveform design, we demonstrate that the communication and sensing functionalities exhibit opposite requirements for waveform amplitude, while converging in their demand for waveform correlation. Regarding waveform amplitude, the communication functionality prefers a random amplitude to convey more information, while the sensing functionality prefers a constant-modulus waveform to guarantee a stable parameter estimation. Regarding waveform correlation, both functionalities prefer a low-correlation waveform with random phases, which not only improves communication efficiency but also provides more independent measurements of sensing parameters.

        \item By characterizing the CMI-SMI and CMI-MSE regions, we reveal the impact of time-frequency resource allocation on the communication-sensing trade-off. Specifically, when the time-frequency resources allocated to communication functionality are doubled, the amount of obtained communication information also doubles. In contrast, doubling these resources for sensing functionality induces a $50\%$ reduction in the MSE.
        
\end{itemize}

\par The remainder of this paper is organized as follows. Section~\ref{section2} introduces the system model of the band-limited continuous-time ISAC system. Section~\ref{section3} presents the performance analysis of the ISAC system, in which the communication-sensing performance metrics are derived and two communication-sensing performance regions are proposed. Section~\ref{section4} provides discussions on the ISAC waveform design and time-frequency-spatial resource allocation. Numerical results are provided in Section~\ref{section5} to verify our analysis. Finally, Section~\ref{section6} concludes this paper.

\emph{{Notations:}} Throughout this paper, the boldface upper/lower case represents matrices/vectors. $( \cdot ) ^{\mathrm{T}}$ and $( \cdot ) ^{\mathrm{H}}$ stand for transpose and Hermitian transpose, respectively.  $\mathbb{E} \{ \cdot \} $ denotes the expected value function. $\mathcal{C} \mathcal{N} ( 0,\sigma ^2 ) $ denotes the complex Gaussian distribution with mean $0$ and variance $\sigma ^2$.  For matrices,
$[ \cdot ] _{ij}$ represents the $(i,j)$-th element, $\mathrm{vec}\left( \mathbf{A} \right) $ denotes the column-stacked vector of $\mathbf{A}$, $\mathbf{A}\otimes \mathbf{B}$ represents the Kronecker product between $\mathbf{A}$ and $\mathbf{B}$. For vectors, $[ \cdot ] _i$ denotes the $i$-th entry. In addition, $\mathrm{card}\left( \mathcal{A} \right) $ is the cardinal number of the finite set $\mathcal{A} $.
\section{System Model}\label{section2}
\begin{figure}[htb]
  \centering
  \includegraphics[width=3.2in]{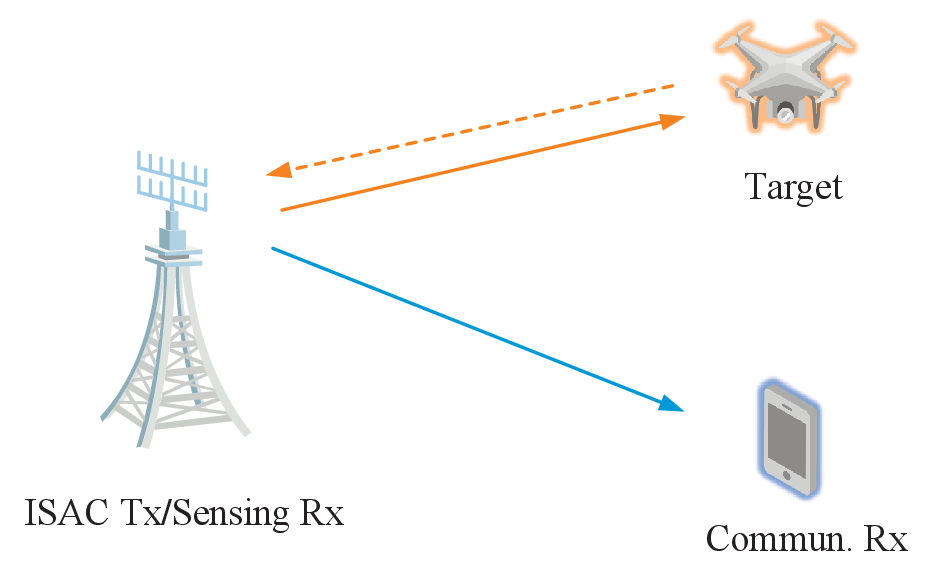}
  \caption{The ISAC scenario considered in this paper.}
  \label{scenario}
\end{figure}
We consider an ISAC system as shown in Fig.~{\ref{scenario}}, in which an ISAC transmitter (Tx) emits ISAC signals to conduct downlink communication and target sensing, simultaneously. The number of antennas at the ISAC Tx, communication receiver (Rx), and sensing Rx are $N$, $M_{\mathrm{c}}$, and $M_{\mathrm{s}}$, respectively. We consider a coherent processing interval (CPI) with $N_{\text{CPI}}$ ISAC symbols. Moreover, sensing parameters and communication channels remain constant during each CPI. 
\subsection{Signal Model}
\par Based on the frequency domain sampling theorem, each ISAC symbol with the bandwidth of $B$ and the duration of $T=1/B$ can be uniquely recovered from its $2B$ frequency-domain samples \cite{tse2005fundamentals}. Hence, the $N_{\text{CPI}}$ ISAC symbols transmitted from the $n$-th antenna of ISAC Tx during one CPI can be described by
\begin{align}
&\tilde{\mathbf{x}}\left( n \right) =
\\
&\left[ \check{\mathbf{x}}_{1}^{\mathrm{T}}\left( n \right) ,\cdots ,\check{\mathbf{x}}_{m}^{\mathrm{T}}\left( n \right) ,\cdots \check{\mathbf{x}}_{N_{\mathrm{CPI}}}^{\mathrm{T}}\left( n \right) \right] ^{\mathrm{T}}\in \mathbb{C} ^{2BN_{\mathrm{CPI}}\times 1},\notag
\end{align}
where $\mathbf{x}_m\left( n \right) \in \mathbb{C} ^{2B\times 1}$ denotes the frequency-domain sample vector of the $m$-th ISAC symbol. 
The ISAC symbols transmitted via $N$ transmit antennas can be characterized by a $2BN_{\mathrm{CPI}}\times N$ matrix as
\begin{align}
\mathbf{X}=\left[ \tilde{\mathbf{x}}\left( 1 \right) ,\cdots ,\tilde{\mathbf{x}}\left( n \right) ,\cdots ,\tilde{\mathbf{x}}\left( N \right) \right] \in \mathbb{C} ^{2BN_{\mathrm{CPI}}\times N}.
\end{align}
\par The receive symbol matrix at the communication Rx can be characterized by 
\begin{align}
\mathbf{Y}_{\mathrm{c}}=\mathbf{XH}_{\mathrm{c}}+\mathbf{N}_{\mathrm{c}}\in \mathbb{C} ^{2BN_{\mathrm{CPI}}\times M_{\mathrm{c}}},
\end{align}
where $\mathbf{H}_{\mathrm{c}}\in \mathbb{C} ^{N\times M_{\mathrm{c}}}$ denotes the communication channel matrix, and $\mathbf{N}_{\mathrm{c}}\in \mathbb{C} ^{2BN_{\mathrm{CPI}}\times M_{\mathrm{c}}}$ is the additive white Gaussian noise (AWGN), whose elements follow the complex Gaussian distribution $\mathcal{C} \mathcal{N} \left( 0,\sigma _{\mathrm{nc}}^{2} \right) $. 
\par Also, the target echo symbol matrix at the sensing Rx is characterized by 
\begin{align}
\mathbf{Y}_{\mathrm{s}}=\mathbf{XH}_{\mathrm{s}}\left( \mathbf{s} \right) +\mathbf{N}_{\mathrm{s}}\in \mathbb{C} ^{2BN_{\mathrm{CPI}}\times M_{\mathrm{s}}},
\end{align}
where $\mathbf{H}_{\mathrm{s}}\left( \mathbf{s} \right) \in \mathbb{C} ^{N\times M_{\mathrm{s}}}$ represents the sensing channel matrix, which involves the sensing parameters $\mathbf{s}$ to be estimated, and $\mathbf{N}_{\mathrm{s}}\in \mathbb{C} ^{2BN_{\mathrm{CPI}}\times M_{\mathrm{s}}}$ is the AWGN, whose elements follow $\mathcal{C} \mathcal{N} \left( 0,\sigma _{\mathrm{ns}}^{2} \right) $. In addition, $\mathbf{s}\in \mathbb{R} ^{K\times 1}$ with $K$ being the parameter dimension.

\section{Communication-Sensing Performance Analysis} \label{section3}
In this section, we analyze the communication-sensing performance of the ISAC system in a unified manner. Specifically, we first characterize the communication performance with the CMI. Then, to unify the performance metric, we characterize the sensing performance from the information-theoretical perspective and propose the SMI. Besides, we reveal the relationship between the SMI and the estimation-theoretic metric, i.e., MSE. Finally, we define the CMI-SMI and CMI-MSE regions to investigate the boundary of communication-sensing performance and the trade-off between the two functionalities.
\subsection{Communication Performance Characterization}
The objective of communication is to recover data information contained in the transmitted signal $\mathbf{X}$ from the received signal $\mathbf{Y}_{\mathrm{c}}$ of the communication Rx. To this end, the communication channel $\mathbf{H}_{\mathrm{c}}$ is usually assumed to be known by the communication Rx so that the data information can be successfully decoded. Therefore, we define the CMI obtained during one CPI as 
\begin{align}
I_{\mathrm{c}}=I\left( \mathbf{X};\mathbf{Y}_{\mathrm{c}}\!\:|\!\:\mathbf{H}_{\mathrm{c}} \right) ,
\end{align}
where $I\left( \mathbf{X};\mathbf{Y}_{\mathrm{c}}\!\:|\!\:\mathbf{H}_{\mathrm{c}} \right) $ denotes the mutual information between $\mathbf{Y}_{\mathrm{c}}$ and $\mathbf{X}$ conditioned on $\mathbf{H}_{\mathrm{c}}$.
\par Let $\mathbf{x}\left( i \right) $, $\mathbf{y}_{\mathrm{c}}\left( i \right) $, and $\mathbf{n}_{\mathrm{c}}\left( i \right) $ denote the $i$-th ($i=1,\cdots ,2BN_{\mathrm{CPI}}$) row of $\mathbf{X}$, $\mathbf{Y}_{\mathrm{c}}$, and $\mathbf{N}_{\mathrm{c}}$, respectively. We consider that each $\mathbf{x}\left( i \right) $/$\mathbf{y}_{\mathrm{c}}\left( i \right) $ in $\mathbf{X}$/$\mathbf{Y}_{\mathrm{c}}$ is independently and identically distributed (i.i.d.). Then, by omitting the index $i$, the CMI can be calculated as
\begin{align}\label{ComMI0}
I_{\mathrm{c}}&=2BN_{\mathrm{CPI}}I\left( \mathbf{x};\mathbf{y}_{\mathrm{c}}\,|\,\mathbf{H}_{\mathrm{c}} \right) 
\\
&=2BN_{\mathrm{CPI}}\left( H\left( \mathbf{y}_{\mathrm{c}}\,|\,\mathbf{H}_{\mathrm{c}} \right) -H\left( \mathbf{y}_{\mathrm{c}}\,|\,\mathbf{x},\mathbf{H}_{\mathrm{c}} \right) \right) \notag
\\
&=2BN_{\mathrm{CPI}}\left( H\left( \mathbf{y}_{\mathrm{c}}\,|\,\mathbf{H}_{\mathrm{c}} \right) -H\left( \mathbf{n}_{\mathrm{c}} \right) \right) \notag
\\
&=BN_{\mathrm{CPI}}I_{\mathrm{c},\mathrm{RE}},\notag
\end{align}
where $I_{\mathrm{c},\mathrm{RE}}\triangleq 2\left( H\left( \mathbf{y}_{\mathrm{c}}\,|\,\mathbf{H}_{\mathrm{c}} \right) -H\left( \mathbf{n}_{\mathrm{c}} \right) \right) $ denotes the CMI of one resource element (RE), which is defined as the time-frequency resource that occupies unit bandwidth and spans one symbol.

\subsection{Sensing Performance Characterization}
The objective of sensing is to extract the sensing parameters of interest $\mathbf{s}$ from the target echo signals $\mathbf{Y}_{\mathrm{s}}$, given the transmitted signals $\mathbf{X}$. Therefore, we define the SMI obtained during one CPI as 
\begin{align}
I_{\mathrm{s}}=I(\mathbf{s};\mathbf{Y}_{\mathrm{s}}\,|\mathbf{X}).
\end{align}
\par To explicitly characterize the SMI, let $\mathbf{h}_{\mathrm{s}}\triangleq \mathrm{vec}\left( \mathbf{H}_{\mathrm{s}} \right) $, and assume that $\mathbf{h}_{\mathrm{s}}$ follows zero-mean Gaussian distributed with an invertible statistical covariance matrix $\mathbf{R}_{\text{H}}$, i.e., $\mathbf{h}_{\text{s}}\thicksim \mathcal{C}\mathcal{N}\left( \mathbf{0},\mathbf{R}_{\text{H}} \right) $. Then, let $\bar{\mathbf{y}}_{\mathrm{s}}=\mathrm{vec}\left( \mathbf{Y}_{\mathrm{s}} \right) $ and $\bar{\mathbf{n}}_{\mathrm{s}}=\mathrm{vec}\left( \mathbf{N}_{\mathrm{s}} \right) $, and the received echo signal at the sensing Rx can be rewritten as
\begin{align}
\bar{\mathbf{y}}_{\mathrm{s}}=\mathcal{X} \!\:\!\:\mathbf{h}_{\mathrm{s}}+\bar{\mathbf{n}}_{\mathrm{s}},
\end{align}
where $\mathcal{X} =\mathbf{X}\otimes \mathbf{I}_{M_{\mathrm{s}}}\in \mathbb{C} ^{2BN_{\mathrm{CPI}}M_{\mathrm{s}}\times M_{\mathrm{s}}N}$ is a auxiliary martix.
\par As such, the SMI can be rewritten as
\begin{align}
I_{\mathrm{s}}&=I(\mathbf{s};\bar{\mathbf{y}}_{\mathrm{s}}\,|\,\mathcal{X} )
=H\left( \bar{\mathbf{y}}_{\mathrm{s}}\!\:|\!\:\mathcal{X} \right) -H\left( \bar{\mathbf{y}}_{\mathrm{s}}\!\:|\mathbf{s},\mathcal{X} \right) .
\end{align}
For the first term $H\left( \bar{\mathbf{y}}_{\mathrm{s}}\!\:|\!\:\mathcal{X} \right) $, it can be easily verified that $\bar{\mathbf{y}}_{\mathrm{s}}$ conditioned on $\mathcal{X}$ is Gaussian distributed with mean $\mathbf{0}$ and covariance $ \mathcal{X} \mathbf{R}_{\mathrm{H}}\mathcal{X} ^H+\sigma _{\mathrm{ns}}^{2}\mathbf{I}_{2BN_{\mathrm{CPI}}M_{\mathrm{s}}}$, i.e.,
\begin{align}
\bar{\mathbf{y}}_{\mathrm{s}}|\mathcal{X} \sim \mathcal{C} \mathcal{N} \left( 0,\mathcal{X} \mathbf{R}_{\mathrm{H}}\mathcal{X} ^H+\sigma _{\mathrm{ns}}^{2}\mathbf{I}_{2BN_{\mathrm{CPI}}M_{\mathrm{s}}} \right) .
\end{align}
Therefore, the conditional entropy $H\left( \bar{\mathbf{y}}_{\mathrm{s}}\!\:|\!\:\mathcal{X} \right) $ can be calculated as
\begin{align}
H\left( \bar{\mathbf{y}}_{\mathrm{s}}\!\:|\!\:\mathcal{X} \right) =\log \left( \det \left( \mathcal{X} \mathbf{R}_{\mathrm{H}}\mathcal{X} ^H+\sigma _{\mathrm{ns}}^{2}\mathbf{I}_{2BN_{\mathrm{CPI}}M_{\mathrm{s}}} \right) \right) ,
\end{align}
and the SMI can be written as 
\begin{align}\label{SensingMI_general}
I_{\mathrm{s}}=&\log \left( \det \left( \mathcal{X} \mathbf{R}_{\mathrm{H}}\mathcal{X} ^H+\sigma _{\mathrm{ns}}^{2}\mathbf{I}_{2BN_{\mathrm{CPI}}M_{\mathrm{s}}} \right) \right) \notag
\\
&-H\left( \bar{\mathbf{y}}_{\mathrm{s}}\!\:|\mathbf{s},\mathcal{X} \right) .
\end{align}
\par For the second term $H\left( \bar{\mathbf{y}}_{\mathrm{s}}\!\:|\mathbf{s},\mathcal{X} \right)$, it relies on the correlation between $\mathbf{s}$ and $\mathbf{h}_{\mathrm{s}}$. As the correlation between $\mathbf{s}$ and $\mathbf{h}_{\mathrm{s}}$ becomes stronger, the conditional entropy $H\left( \bar{\mathbf{y}}_{\mathrm{s}}\!\:|\mathbf{s},\mathcal{X} \right)$ decreases, and consequently the SMI gets larger. 

\par To establish a bridge between the proposed SMI and the traditional estimation-theoretic metric, we denote $\varepsilon$ as the minimum average MSE of all sensing parameters to be estimated. Then, combining the detection\&estimation theory and the information theory, $\varepsilon$ can be provided by the following Proposition.
\begin{proposition}\label{proposition1}
With SMI being $I_{\mathrm{s}}$ and the auto-covariance matrix of $\mathbf{s}$ being $\mathbf{R}_{\mathrm{s}}$, the MSE $\varepsilon$ is given by
\begin{align}\label{SensingMSE}
\varepsilon =2^{\frac{1}{K}\left( \log\det\mathrm{(}\mathbf{R}_{\mathrm{s}})-I_{\mathrm{s}} \right)}.
\end{align}
\end{proposition}
\begin{proof}
\par According to the data processing inequality feature of information theory, it readily follows that 
\begin{align}\label{inequality1}
I_{\mathrm{s}}&=I\left( \mathbf{s};\mathbf{Y}_{\mathrm{s}}\,|\mathbf{X} \right) \notag
\\
&\ge I\left( \mathbf{s};\hat{\mathbf{s}}\,|\,\mathbf{X} \right) \notag
\\
&=H\left( \mathbf{s}\,|\,\mathbf{X} \right) -H\left( \mathbf{s}\,|\,\hat{\mathbf{s}},\mathbf{X} \right) ,
\end{align}
where $H\left( \mathbf{s}\,|\,\mathbf{X} \right) $ denotes the conditional entropy of $\mathbf{s}$ given $\mathbf{X}$ and $H\left( \mathbf{s}\,|\,\hat{\mathbf{s}},\mathbf{X} \right) $ denotes the conditional entropy of $\mathbf{s}$ given $\mathbf{X}$ and the estimated $\mathbf{\hat{s}}$. 
We denote the estimation error as $\mathbf{\boldsymbol{\epsilon } }=\mathbf{s}-\mathbf{\hat{s}}$, and consider that the estimation error is Gaussian distributed with an invertible statistical covariance matrix $\mathbf{R}_{\boldsymbol{\epsilon }}$, i.e., 
$\mathbf{\boldsymbol{\epsilon }}\thicksim\mathcal{C}\mathcal{N}\left(\mathbf{0},\mathbf{R}_{\boldsymbol{\epsilon }} \right) $. Then, we have
\begin{align}
H\left( \mathbf{s}\,|\,\hat{\mathbf{s}},\mathbf{X} \right) &=H\left( \mathbf{\boldsymbol{\epsilon } }\,|\,\mathbf{X} \right) =K\log 2\pi e+\log\det\mathrm{(}\mathbf{R}_{\boldsymbol{\epsilon }}),
\\
H\left( \mathbf{s}\,|\,\mathbf{X} \right) &=K\log 2\pi e+\log\det\mathrm{(}\mathbf{R}_{\mathrm{s}}),
\end{align}
where $\mathbf{R}_{\mathrm{s}}$ denotes the auto-covariance matrix of $\mathbf{s}$. 
\par We assume that elements in $\mathbf{\boldsymbol{\epsilon }}$ are independent with each other and the estimator is unbiased. As such, (\ref{inequality1}) can be rewritten as
\begin{align}
I\left( \mathbf{s};\mathbf{Y}_{\mathrm{s}}\,|\mathbf{X} \right)  &\ge H\left( \mathbf{s}\,|\,\mathbf{X} \right) -H\left( \mathbf{\boldsymbol{\epsilon } }\,|\,\mathbf{X} \right) 
\\
&=\log\det\mathrm{(}\mathbf{R}_{\mathrm{s}})-\log\det\mathrm{(}\mathbf{R}_{\boldsymbol{\epsilon }})\notag
\\
&=\log\det\mathrm{(}\mathbf{R}_{\mathrm{s}})-\sum_{k=1}^{K}{\log \mathbb{E} \{\epsilon _{k}^{2}\}},\notag
\end{align}
and $\mathbb{E} \{\epsilon _{k}^{2}\}$ is equal to the MSE of the $k$-th element of $\mathbf{s}$. The above equation can also be written as
\begin{align}
\frac{1}{K}\left( \log\det\mathrm{(}\mathbf{R}_{\mathrm{s}})-I(\mathbf{s};\mathbf{Y}_{\mathrm{s}}\,|\mathbf{X}) \right) &\le \log \left( \prod_{k=1}^K{\mathbb{E} \{\epsilon _{k}^{2}\}} \right) ^{\frac{1}{K}}\notag
\\
&\!\!\!\!\!\!\le \log \left( \frac{1}{K}\sum_{k=1}^K{\mathbb{E} \{\epsilon _{k}^{2}\}} \right). 
\end{align}
\par Hence, the average MSE of all sensing parameters to be estimated (i.e., $\frac{1}{K}\sum_{k=1}^K{\mathbb{E} \{\epsilon _{k}^{2}\}}$) satisfies
\begin{align}\label{MSE}
\frac{1}{K}\sum_{k=1}^K{\mathbb{E} \{\epsilon _{k}^{2}\}}\ge 2^{\frac{1}{K}\left( \log\det\mathrm{(}\mathbf{R}_{\mathrm{s}})-I(\mathbf{s};\mathbf{Y}_{\mathrm{s}}\,|\mathbf{X}) \right)},
\end{align}
where the equality holds if and only if no information is lost during data processing and the MSE of all sensing parameters are the same, at which point $\frac{1}{K}\sum_{k=1}^K{\mathbb{E} \{\epsilon _{k}^{2}\}}$ reaches its minimum value, which is defined as the MSE.
\end{proof}
\par Proposition~\ref{proposition1} indicates that acquiring more SMI can lower the MSE. Specifically, the MSE decreases exponentially as the SMI increases.

\subsection{Communication-Sensing Performance Region}
We propose the communication-sensing performance region to investigate the boundary of communication-sensing performance and reveal the trade-off between the two functionalities. 
\par First, let $\mathcal{U} _{\mathrm{ISAC}}$ denote the available REs for the ISAC system during each CPI, and let $\mathcal{U} _{\mathrm{c}}$ and $\mathcal{U} _{\mathrm{s}}$ respectively denote the REs allocated to communication and sensing functionalities, satisfying $\mathcal{U} _{\mathrm{c}}\cup \mathcal{U} _{\mathrm{s}}=\mathcal{U} _{\mathrm{ISAC}}$. Then, given $\mathcal{U} _{\mathrm{ISAC}}$, we define the CMI-SMI region and the CMI-MSE region as the closure of the set of all achievable pairs $\left( I_{\mathrm{c}},I_{\mathrm{s}} \right) $ and $\left( I_{\mathrm{c}},\varepsilon  \right) $, respectively, which are given by
\begin{align}
&\mathcal{C} _{CMI-SMI}=\left\{ \left( I_{\mathrm{c}}\left( \mathcal{U} _{\mathrm{c}} \right) ,I_{\mathrm{s}}\left( \mathcal{U} _{\mathrm{s}} \right) \right) | \,\mathcal{U} _{\mathrm{c}}\cup \mathcal{U} _{\mathrm{s}}=\mathcal{U} _{\mathrm{ISAC}} \right\} ,
\\
&\mathcal{C} _{CMI-MSE}=\left\{ \left( I_{\mathrm{c}}\left( \mathcal{U} _{\mathrm{c}} \right) ,\varepsilon \left( \mathcal{U} _{\mathrm{s}} \right) \right) | \,\mathcal{U} _{\mathrm{c}}\cup \mathcal{U} _{\mathrm{s}}=\mathcal{U} _{\mathrm{ISAC}} \right\} .
\end{align}
\par For communication, (\ref{ComMI0}) shows that the CMI is linearly positively correlated to the time-frequency resources. For sensing, intuitively, allocating more time-frequency resources will increase the SMI and decrease the MSE. Hence, there exists a trade-off in time-frequency resource allocation, which is in accordance with the typical CMI-SMI and CMI-MSE regions as illustrated in Fig.~\ref{CI-SI} and Fig.~\ref{CI-SM}.

\begin{figure}[htb]
  \centering
  \subfigure[CMI-SMI region.]
  {
  \label{CI-SI}
  \includegraphics[width=2.8in]{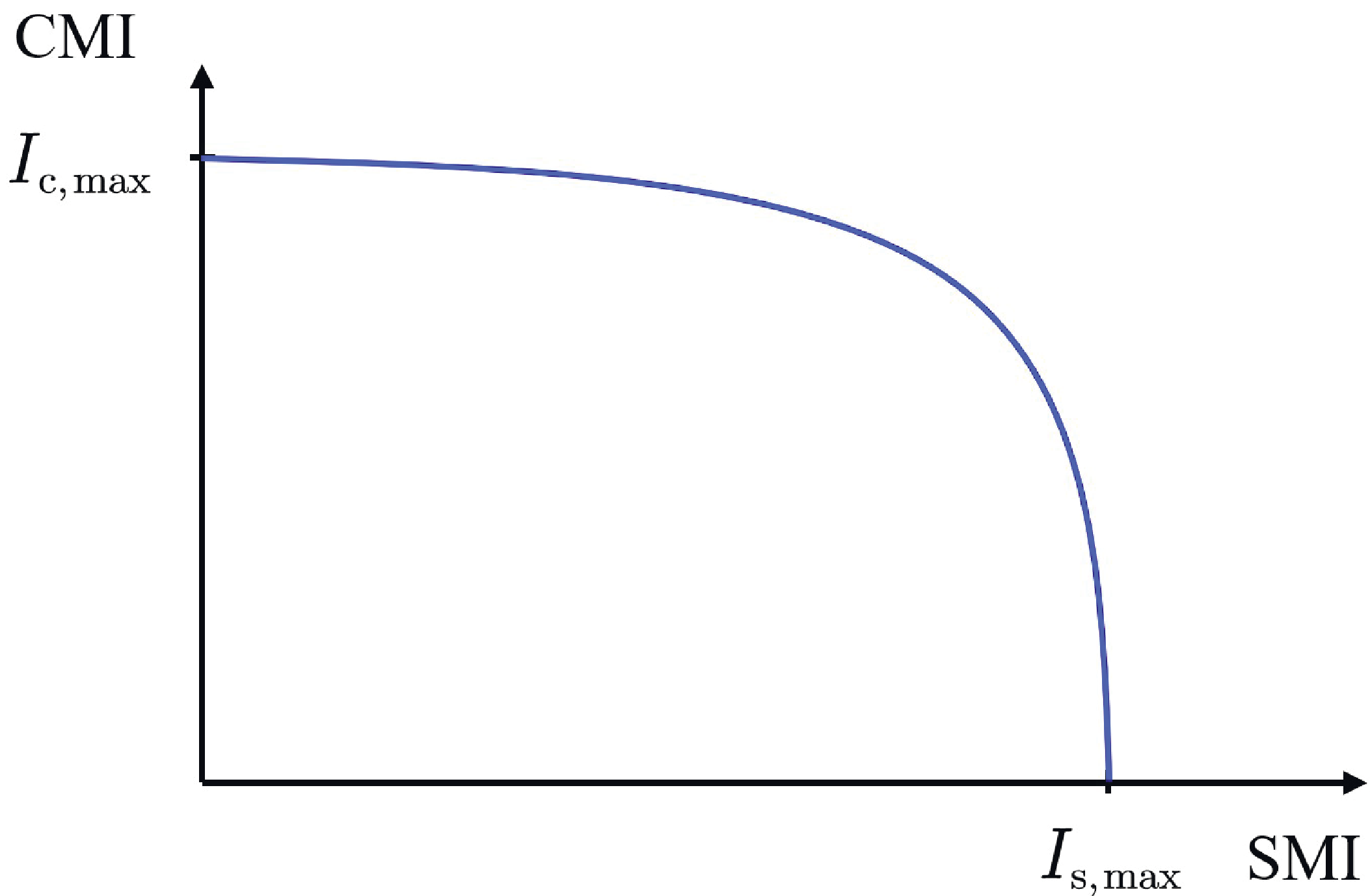}
  }
  ~\subfigure[CMI-MSE region.]
  {
  \label{CI-SM}
  \includegraphics[width=2.8in]{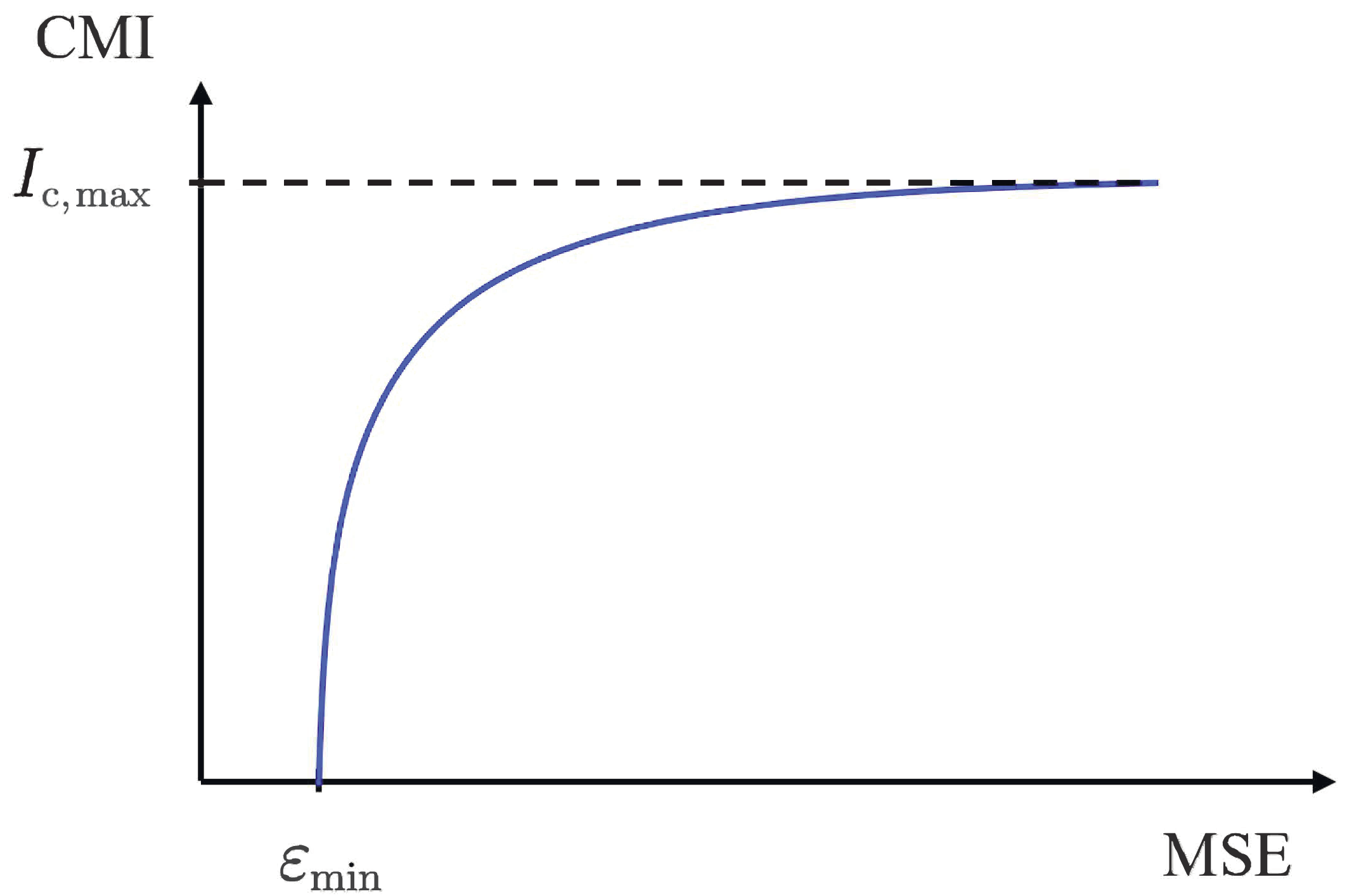}
  }
  \caption{Illustration of the communication-sensing performance region.}
\end{figure}

\section{Special Cases and Discussion} \label{section4}
In this section, we consider the special case where sensing parameters to be estimated are partial elements of the sensing
channel. Specifically, we first analyze the SMI and MSE for sensing channel estimation. Then, we investigate the ISAC waveform design, demonstrating that the sensing functionality prefers constant-modulus waveform with low correlation, while the communication functionality prefers Gaussian waveform. Next, we characterize the CMI-SMI and CMI-MSE regions, revealing the impacts of resource allocation on the communication-sensing trade-off. Specifically, when the time-frequency resources allocated to communication functionality are doubled, the amount of obtained communication information also doubles. In contrast, doubling these resources for sensing functionality leads to a $50\%$ reduction in the MSE.
\par We consider that the sensing parameters $\mathbf{s}$ to be estimated are partial elements of the sensing channel, i.e., $\mathbf{s}\subseteq \mathbf{h}_{\mathrm{s}}$,
and denote the remaining parameters of $ \mathbf{h}_{\mathrm{s}}$ excluding $\mathbf{s}$ as $\mathbf{r}\in \mathbb{C} ^{\left( NM_{\mathrm{s}}-K \right) \times 1}$. In the following, to characterize the SMI in this scenario, we focus on the calculation of the $H\left( \bar{\mathbf{y}}_{\mathrm{s}}\!\:|\mathbf{s},\mathcal{X} \right) $ term in (\ref{SensingMI_general}).
\par First, it can be easily verified that the remaining parameters $\mathbf{r}$ given sensing parameters $\mathbf{s}$ follow a conditional Gaussian distribution with mean $\mathbf{0}$ and conditional covariance matrix $\mathbf{R}_{\mathrm{r}|\mathrm{s}}$, i.e., 
\begin{align}
\mathbf{r}|\mathbf{s}\thicksim \mathcal{C} \mathcal{N} \left( \mathbf{0},\mathbf{R}_{\mathrm{r}|\mathrm{s}} \right) ,
\end{align}
where the conditional covariance matrix $\mathbf{R}_{\mathrm{r}|\mathrm{s}}$ is given by
\begin{align}
\mathbf{R}_{\mathrm{r}|\mathrm{s}}=\mathbf{R}_{\mathrm{r}}-\mathbf{R}_{\mathrm{rs}}\mathbf{R}_{\mathrm{s}}^{-1}\mathbf{R}_{\mathrm{sr}},
\end{align}
where $\mathbf{R}_{\mathrm{r}}$ denotes the auto-covariance matrix of $\mathbf{r}$, $\mathbf{R}_{\mathrm{rs}}$ and $\mathbf{R}_{\mathrm{sr}}$ are the cross-covariance matrices between $\mathbf{s}$ and $\mathbf{r}$. Let
$\left[ \mathbf{r} \right] _{k_{\mathrm{o}}}=\left[ \mathbf{h}_{\mathrm{s}} \right] _{\psi \left( k_{\mathrm{o}} \right)},k_{\mathrm{o}}=1,\cdots ,\left( NM_s-K \right) $, where $\psi \left( k_{\mathrm{o}} \right) $ is a mapping function. Then, the channel vector $\mathbf{h}_{\mathrm{s}}$ given sensing parameters $\mathbf{s}$ follows a conditional Gaussian distribution with mean $\mathbf{0}$ and conditional covariance matrix $\mathbf{R}_{\mathbf{h}_{\mathrm{s}}|\mathbf{s}}$, i.e.,
\begin{align}
\mathbf{h}_{\mathrm{s}}|\mathbf{s}\thicksim \mathcal{C} \mathcal{N} \left(\mathbf{0} ,\mathbf{R}_{\mathbf{h}_{\mathrm{s}}|\mathbf{s}} \right) ,
\end{align}
where $\left[ \mathbf{R}_{\mathbf{h}_{\mathrm{s}}|\mathbf{s}} \right] _{\psi (i),\psi (j)}\!=\!\left[ \mathbf{R}_{\mathrm{r}|\mathrm{s}} \right] _{i,j},i,j=1,\cdots ,\left( NM_s\!-\!K \right) $, and other elements in $\mathbf{R}_{\mathbf{h}_{\mathrm{s}}|\mathbf{s}}$ are $0$.  As such, we have
\begin{align}
\left( \bar{\mathbf{y}}_{\mathrm{s}}|\mathbf{s},\mathcal{X} \right) \thicksim \mathcal{C} \mathcal{N} \left( \mathbf{0},\mathcal{X} \mathbf{R}_{\mathbf{h}_{\mathrm{s}}|\mathbf{s}}\mathcal{X} ^H+\sigma _{\mathrm{ns}}^{2}\mathbf{I}_{2BN_{\mathrm{CPI}}M_{\mathrm{s}}} \right) .
\end{align}
\par Based on the above, the SMI can be calculated as
\begin{align}\label{SensingMI_eg1}
I_{\mathrm{s}}=&\log \left( \det \left( \mathcal{X} \mathbf{R}_{\mathrm{H}}\mathcal{X} ^H+\sigma _{\mathrm{ns}}^{2}\mathbf{I}_{2BN_{\mathrm{CPI}}M_{\mathrm{s}}} \right) \right) \notag
\\
&-\log \left( \det \left( \mathcal{X} \mathbf{R}_{\mathbf{h}_{\mathrm{s}}|\mathbf{s}}\mathcal{X} ^H+\sigma _{\mathrm{ns}}^{2}\mathbf{I}_{2BN_{\mathrm{CPI}}M_{\mathrm{s}}} \right) \right) ,
\end{align}
and the average MSE (\ref{SensingMSE}) can be calculated as (\ref{average MSE}) at the top of this page.
\begin{figure*}[ht]
\begin{equation}\label{average MSE}
 \begin{gathered}
  \begin{aligned}
\varepsilon\ge 2^{\frac{1}{K}\left( \log\det\mathrm{(}\mathbf{R}_{\mathrm{s}})-\log \left( \det \left( \mathcal{X} \mathbf{R}_{\mathrm{H}}\mathcal{X} ^H+\sigma _{\mathrm{ns}}^{2}\mathbf{I}_{2BN_{\mathrm{CPI}}M_{\mathrm{s}}} \right) \right) +\log \left( \det \left( \mathcal{X} \mathbf{R}_{\mathbf{h}_{\mathrm{s}}|\mathbf{s}}\mathcal{X} ^H+\sigma _{\mathrm{ns}}^{2}\mathbf{I}_{2BN_{\mathrm{CPI}}M_{\mathrm{s}}} \right) \right) \right)}.
 \end{aligned}
 \end{gathered}
\end{equation}
\hrulefill
\end{figure*}
\par More specially, when $\mathbf{s}=\mathbf{h}_{\mathrm{s}}$, the conditional covariance matrix $\mathbf{R}_{\mathbf{h}_{\mathrm{s}}|\mathbf{s}}$ becomes a null matrix, and thus the SMI in (\ref{SensingMI_eg1}) becomes
\begin{align}\label{SensingMI_eg2}
I_{\mathrm{s}}&=\log \left( \frac{\det \left( \mathcal{X} \mathbf{R}_{\mathrm{H}}\mathcal{X} ^H+\sigma _{\mathrm{ns}}^{2}\mathbf{I}_{2BN_{\mathrm{CPI}}M_{\mathrm{s}}} \right)}{\det \left( \sigma _{\mathrm{ns}}^{2}\mathbf{I}_{2BN_{\mathrm{CPI}}M_{\mathrm{s}}} \right)} \right) 
\\
&=\log \left( \det \left( \sigma _{\mathrm{ns}}^{-2}\mathcal{X} \mathbf{R}_{\mathrm{H}}\mathcal{X} ^H+\mathbf{I}_{2BN_{\mathrm{CPI}}M_{\mathrm{s}}} \right) \right). \notag
\end{align}
Based on Proposition~\ref{proposition1}, the MSE (\ref{SensingMSE}) can be calculated as
\begin{align}
\varepsilon\ge 2^{\frac{1}{K}\left( \log\det\mathrm{(}\mathbf{R}_{\mathrm{s}})-\log \left( \det \left( \sigma _{\mathrm{ns}}^{-2}\mathcal{X} \mathbf{R}_{\mathrm{H}}\mathcal{X} ^H+\mathbf{I}_{2BN_{\mathrm{CPI}}M_{\mathrm{s}}} \right) \right) \right)}.
\end{align}

\subsection{Waveform Design in ISAC system}
\begin{proposition}\label{proposition2}
In the design of ISAC waveforms, communication and sensing functionalities exhibit opposite requirements in terms of waveform amplitude, while they converge in their demand for waveform correlation. In terms of waveform amplitude, the communication functionality prefers a more random amplitude to convey more information, while the sensing functionality prefers a constant-modulus waveform to ensure a stable estimation for sensing parameters. In terms of waveform correlation, both two functionalities prefer a low-correlation waveform with random phases, which can not only improve communication efficiency but also attain more independent measurements of sensing parameters. 
\end{proposition}
\begin{proof}
To simplify the analysis, we consider $\mathbf{s}=\mathbf{h}_{\mathrm{s}}$ and $\mathbf{R}_{\mathrm{H}}=\alpha _{\mathrm{H}_{\mathrm{s}}}^{2}\mathbf{I}_{M_{\mathrm{s}}N}$, where $\alpha _{\mathrm{H}_{\mathrm{s}}}^{2}$ denotes the channel gain. Let $\beta _{\mathrm{s}}\triangleq \alpha _{\mathrm{H}_{\mathrm{s}}}^{2}/\sigma _{\mathrm{ns}}^{2}$, and the SMI (\ref{SensingMI_eg2}) can be rewritten as
\begin{align}
I_{\mathrm{s}}&\overset{\left( a \right)}{=}\log \left( \det \left( \beta _{\mathrm{s}}\mathcal{X} ^H\mathcal{X} +\mathbf{I}_{M_{\mathrm{s}}N} \right) \right) 
\\
&=\log \left( \det \left( \beta _{\mathrm{s}}\left( \mathbf{X}^H\mathbf{X}\otimes \mathbf{I}_{M_{\mathrm{s}}} \right) +\mathbf{I}_{M_{\mathrm{s}}N} \right) \right) \notag
\\
&\triangleq \log \left( \det \left( \beta _{\mathrm{s}}\left( \check{\mathbf{X}}\otimes \mathbf{I}_{M_{\mathrm{s}}} \right) +\mathbf{I}_{M_{\mathrm{s}}N} \right) \right) ,\notag
\end{align}
where $\left( a \right) $ holds due to
\begin{align}
\det \left( \mathbf{I}_{\mathrm{rA}}+\mathbf{AB} \right) =\det \left( \mathbf{I}_{\mathrm{rB}}+\mathbf{BA} \right),
\end{align}
and $\check{\mathbf{X}}\triangleq \mathbf{X}^H\mathbf{X}\in \mathbb{C} ^{N\times N}$.
\par In the following, we will analyze the characteristics of optimal sensing waveform in two cases.
\begin{itemize}
    \item $N=1$: In this case, the average SMI is given by
\begin{align}
\mathbb{E} _{\mathbf{X}}\left\{ I\left( \mathbf{s};\bar{\mathbf{y}}_{\mathrm{s}}\,|\,\mathbf{X} \right) \right\} &=\mathbb{E} _{\mathbf{X}}\left\{ \log \left( \det \left( \left( \beta _{\mathrm{s}}\check{\mathrm{X}}+1 \right) \mathbf{I}_{M_{\mathrm{s}}} \right) \right) \right\} \notag
\\
&=M_{\mathrm{s}}\mathbb{E} _{\mathbf{X}}\left\{ \log \left( \beta _{\mathrm{s}}\check{\mathrm{X}}+1 \right) \right\} \notag
\\
&\overset{\left( b \right)}{\le}M_{\mathrm{s}}\log \left( \beta _{\mathrm{s}}\mathbb{E} _{\mathbf{X}}\left\{ \check{\mathrm{X}} \right\} +1 \right) ,
\end{align}
where (b) follows the Jensen inequality, and the equality holds if and only if $\mathbb{E} _{\mathbf{X}}\left\{ \check{\mathrm{X}} \right\} =\check{\mathrm{X}}$. Since $\check{\mathrm{X}}=\mathbf{X}^H\mathbf{X}=\sum_{i=1}^{2BN_{\mathrm{CPI}}}{\left| \left[ \mathbf{X} \right] _i \right|^2}$, the equality of (b) holds when $\left| \left[ \mathbf{X} \right] _i \right|$ is a constant. This indicates that the maximum value of the average SMI in the case of $N=1$ will be achieved when the waveform is constant-modulus.

\item $N>1$: In this case, the average SMI is given by
\begin{align}
&\mathbb{E} _{\mathbf{X}}\left\{ I\left( \mathbf{s};\bar{\mathbf{y}}_{\mathrm{s}}\,|\,\mathbf{X} \right) \right\} =
\\
&\mathbb{E} _{\mathbf{X}}\left\{ \log \left( \det \left( \beta _{\mathrm{s}}\left( \check{\mathbf{X}}\otimes \mathbf{I}_{M_{\mathrm{s}}} \right) +\mathbf{I}_{M_{\mathrm{s}}N} \right) \right) \right\} .\notag
\end{align}
According to \cite{YangMIMOradar}, for a positive semi-definite Hermitian matrix $\mathbf{A}$, its determinant is less than or equal to the product of its main diagonal elements, and the equation holds if and only if $\mathbf{A}$ is a diagonal matrix. Hence, the maximum value of $\left| \beta _{\mathrm{s}}\left( \check{\mathbf{X}}\otimes \mathbf{I}_{M_{\mathrm{s}}} \right) +\mathbf{I}_{M_{\mathrm{s}}N} \right|$ will be attained when $\beta _{\mathrm{s}}\left( \check{\mathbf{X}}\otimes \mathbf{I}_{M_{\mathrm{s}}} \right) +\mathbf{I}_{M_{\mathrm{s}}N}$ is diagonal, which indicates that $\check{\mathbf{X}}$ should be a diagonal matrix. Since $\left[ \check{\mathbf{X}} \right] _{m,n}=\sum_{i=1}^{2BN_{\mathrm{CPI}}}{\left[ \mathbf{X} \right] _{i,m}^{H}\left[ \mathbf{X} \right] _{i,n}}$, in order to make $\check{\mathbf{X}}$ approximate the diagonal matrix, each element of $\mathbf{X}$ should be i.i.d. with zero mean. 
\end{itemize}
Hence, the sensing functionality prefers a constant-modulus waveform with low correlation. For communication purposes, under the constraints of given mean and variance, it is widely recognized that the most effective communication waveform conforms to a Gaussian distribution \cite{cover1999elements}.
This indicates that the design of communication and sensing waveforms entails a paradoxical balance, exhibiting both conflict and consistency. For sensing functionality, on the one hand, the waveform should have determinacy to ensure a stable estimation for sensing parameters. On the other hand, to acquire more information about the sensing parameters, the waveform should have a low correlation to attain more independent measurements of sensing parameters. 
\end{proof}

\par Based on Proposition~\ref{proposition2}, the ISAC waveform can be designed to feature low correlation, a random phase, and a deterministic-random adjustable amplitude. It is worth noting that there will exist a deterministic-random trade-off in the design of the ISAC waveform amplitude. The more constant the waveform amplitude is, the more stable parameter estimation can be acquired, thereby yielding enhanced sensing performance. Conversely, when the waveform amplitude is closer to the Rayleigh distribution, the ISAC waveform will be more similar to the Gaussian waveform, which has a stronger ability to convey information, thereby achieving better communication performance. 

\subsection{Time-Frequency-Spatial Resource Allocation in ISAC Systems}
In this subsection, we will investigate the resource allocation of the ISAC system, and begin by giving the SMI and the MSE of the ISAC system in the following Lemma.
\begin{lemma}\label{lemma1}
The SMI and the MSE of the ISAC system can be approximately given by
\begin{align}
&I_{\mathrm{s}}=K\log \left( \frac{2BN_{\mathrm{CPI}}P_t}{\sigma _{\mathrm{ns}}^{2}} \right) +\log \left( \det \left( \mathbf{R}_{\mathrm{s}} \right) \right) ,
\\
&\varepsilon = \frac{\sigma _{\mathrm{ns}}^{2}}{2BN_{\mathrm{CPI}}P_t}.
\end{align}
\end{lemma}
\begin{proof}
Based on subsection A, we assume that each element of $\mathbf{X}$ are i.i.d. with mean zero and variance $P_t$, which yields $\check{\mathbf{X}}\approx 2BN_{\mathrm{CPI}}P_t\mathbf{I}_N$. Hence, the SMI can be rewritten as
\begin{align}
I_{\mathrm{s}}=&\log \left( \det \left( \mathbf{R}_{\mathrm{H}}\left( \check{\mathbf{X}}\otimes \mathbf{I}_{M_{\mathrm{s}}} \right) +\sigma _{\mathrm{ns}}^{2}\mathbf{I}_{M_{\mathrm{s}}N} \right) \right) \notag
\\
&-\log \left( \det \left( \mathbf{R}_{\mathbf{h}_{\mathrm{s}}|\mathbf{s}}\left( \check{\mathbf{X}}\otimes \mathbf{I}_{M_{\mathrm{s}}} \right) +\sigma _{\mathrm{ns}}^{2}\mathbf{I}_{M_{\mathrm{s}}N} \right) \right) \notag
\\
\approx& \log \left( \frac{\det \left( 2BN_{\mathrm{CPI}}P_t\mathbf{R}_{\mathrm{H}}+\sigma _{\mathrm{ns}}^{2}\mathbf{I}_{M_{\mathrm{s}}N} \right)}{\det \left( 2BN_{\mathrm{CPI}}P_t\mathbf{R}_{\mathbf{h}_{\mathrm{s}}|\mathbf{s}}+\sigma _{\mathrm{ns}}^{2}\mathbf{I}_{M_{\mathrm{s}}N} \right)} \right) \notag
\\
=&\log \left( \frac{\det \left( \mathbf{R}_{\mathrm{H}}+\frac{\sigma _{\mathrm{ns}}^{2}}{2BN_{\mathrm{CPI}}P_t}\mathbf{I}_{M_{\mathrm{s}}N} \right)}{\det \left( \mathbf{R}_{\mathbf{h}_{\mathrm{s}}|\mathbf{s}}+\frac{\sigma _{\mathrm{ns}}^{2}}{2BN_{\mathrm{CPI}}P_t}\mathbf{I}_{M_{\mathrm{s}}N} \right)} \right). 
\end{align}
\par Since the magnitude of $\frac{\sigma _{\mathrm{ns}}^{2}}{2BN_{\mathrm{CPI}}P_t}$ is much smaller than the magnitudes of elements in $\mathbf{R}_{\mathrm{H}}$ and $\mathbf{R}_{\mathbf{h}_{\mathrm{s}}|\mathbf{s}}$, we have
\begin{align}
&\det \left( \mathbf{R}_{\mathrm{H}}+\frac{\sigma _{\mathrm{ns}}^{2}}{2BN_{\mathrm{CPI}}P_t}\mathbf{I}_{M_{\mathrm{s}}N} \right) \notag
\\
&\qquad\approx \det \left( \left[ \begin{matrix}
	\mathbf{R}_{\mathrm{s}}&		\mathbf{R}_{\mathrm{sr}}\\
	\mathbf{R}_{\mathrm{rs}}&		\mathbf{R}_{\mathrm{r}}\\
\end{matrix} \right] \right) \notag
\\
&\qquad=\det \left( \mathbf{R}_{\mathrm{s}} \right) \det \left( \mathbf{R}_{\mathrm{r}}-\mathbf{R}_{\mathrm{rs}}\mathbf{R}_{\mathrm{s}}^{-1}\mathbf{R}_{\mathrm{sr}} \right) \notag
\\
&\qquad=\det \left( \mathbf{R}_{\mathrm{s}} \right) \det \left( \mathbf{R}_{\mathrm{r}|\mathrm{s}} \right) ,
\\
&\det \left( \mathbf{R}_{\mathbf{h}_{\mathrm{s}}|\mathbf{s}}+\frac{\sigma _{\mathrm{ns}}^{2}}{2BN_{\mathrm{CPI}}P_t}\mathbf{I} \right) \notag
\\
&\qquad\approx \det \left( \left[ \begin{matrix}
	\frac{\sigma _{\mathrm{ns}}^{2}}{2BN_{\mathrm{CPI}}P_t}\mathbf{I}&		0\\
	0&		\mathbf{R}_{\mathrm{r}|\mathrm{s}}\\
\end{matrix} \right] \right) \notag
\\
&\qquad=\left( \frac{\sigma _{\mathrm{ns}}^{2}}{2BN_{\mathrm{CPI}}P_t} \right) ^K\det \left( \mathbf{R}_{\mathrm{r}|\mathrm{s}} \right) .
\end{align}
Hence, the SMI is approximate to
\begin{align}\label{sensing_MI_conclution}
I_{\mathrm{s}}&=\log \left( \frac{\det \left( \mathbf{R}_{\mathrm{s}} \right) \det \left( \mathbf{R}_{\mathrm{r}|\mathrm{s}} \right)}{\left( \frac{\sigma _{\mathrm{ns}}^{2}}{2BN_{\mathrm{CPI}}P_t} \right) ^K\det \left( \mathbf{R}_{\mathrm{r}|\mathrm{s}} \right)} \right) 
\\
&=K\log \left( \frac{2BN_{\mathrm{CPI}}P_t}{\sigma _{\mathrm{ns}}^{2}} \right) +\log \left( \det \left( \mathbf{R}_{\mathrm{s}} \right) \right),  \notag
\end{align}
and the MSE is approximate to
\begin{align}\label{sensing_MSE_conclution}
\varepsilon &=2^{\frac{1}{K}\left( \log \left( \det \left( \mathbf{R}_{\mathrm{s}} \right) \right) -K\log \left( \frac{2BN_{\mathrm{CPI}}P_t}{\sigma _{\mathrm{ns}}^{2}} \right) -\log \left( \det \left( \mathbf{R}_{\mathrm{s}} \right) \right) \right)}
\\
&=2^{-\log \left( \frac{2BN_{\mathrm{CPI}}P_t}{\sigma _{\mathrm{ns}}^{2}} \right)}=\frac{\sigma _{\mathrm{ns}}^{2}}{2BN_{\mathrm{CPI}}P_t}.\notag
\end{align}
\end{proof}

\par Based on Lemma~\ref{lemma1}, we have the following theorem.
\begin{theorem}\label{theorem1}
The achievable CMI-SMI region and CMI-MSE region with  given ISAC resources $\mathcal{U} _{\mathrm{ISAC}}$ are respectively given by
\begin{align}
&\mathcal{C} _{CMI\!-\!SMI}=\left\{ \left( \mathcal{U} _{\mathrm{c}}I_{\mathrm{c},\mathrm{RE}},I_{\mathrm{s}}\left( \mathcal{U} _{\mathrm{s}} \right) \right) \middle| \!\:\mathcal{U} _{\mathrm{c}}\cup \mathcal{U} _{\mathrm{s}}=\mathcal{U} _{\mathrm{ISAC}} \right\} ,
\\
&\mathcal{C} _{CMI\!-\!SMI}\!=\!\left\{ \left( \mathcal{U} _{\mathrm{c}}I_{\mathrm{c},\mathrm{RE}},\frac{\sigma _{\mathrm{ns}}^{2}}{2\mathcal{U} _{\mathrm{s}}P_t} \right) \middle| \!\:\mathcal{U} _{\mathrm{c}}\cup \mathcal{U} _{\mathrm{s}}\!=\!\mathcal{U} _{\mathrm{ISAC}} \right\} ,
\end{align}
where
\begin{align}
I_{\mathrm{s}}\left( \mathcal{U} _{\mathrm{s}} \right) =K\log \left( \mathcal{U} _{\mathrm{s}} \right) \!+\!K\log \left( \frac{2P_t}{\sigma _{\mathrm{ns}}^{2}} \right) \!+\!\log \left( \det \left( \mathbf{R}_{\mathrm{s}} \right) \right) .
\end{align}
\end{theorem}

\subsubsection{Time-Frequency Resources Allocation in ISAC systems}
As indicated by Theorem~\ref{theorem1}, when the time-frequency resources allocated to the communication functionality double, the amount of the obtained communication information also doubles. When the time-frequency resources allocated to the sensing functionality double, the amount of the obtained information about the sensing parameters increases by $3K$ bits, making the sensing estimation error reduce by $50\%$.
The communication-sensing region indicates that we should make a trade-off when allocating the time-frequency resources for communication and sensing. When the time-frequency resource allocation for one functionality is sufficient, additional resources should be given to another functionality to maximize the performance improvement of the ISAC system.

\subsubsection{Spatial Resources Allocation in ISAC systems}
Lemma ~\ref{lemma1} demonstrates that, for sensing channel estimation, with given sensing parameters of interest $\mathbf{s}$, adding more spatial resources (i.e., increasing the dimension of the remaining parameters $\mathbf{r}$) will not affect the SMI and the MSE. This is because, the sensing information obtained from $\mathbf{r}$ is redundant to that obtained from $\mathbf{s}$, offering negligible extra information about sensing parameters. In this case, additional spatial resources should be allocated for communication functions. In addition, lemma ~\ref{lemma1} also suggests that the MSE is not affected by the spatial correlation between the sensing parameters. This is because, when the sensing parameters are more correlated, on the one hand, the determinant of $\mathbf{R}_{\mathrm{s}}$ becomes smaller, thereby reducing the acquired SMI. On the other hand, more spatial correlation between the sensing parameters increases the gain of joint parameter estimation, which ultimately stabilizes the MSE.

\section{Numerical Results}\label{section5}
In this section, we present numerical results to validate the proposed performance metrics, characterize the proposed communication-sensing performance region, as well as to investigate the impacts of waveforms and other system parameters on the communication-sensing performance.

\subsection{Relationship between SMI and MSE}
\par First, we investigate the relationship between SMI and MSE discussed in Proposition~\ref{proposition1}. 
\begin{figure}[htb]
  \centering
  \includegraphics[width=3.3in]{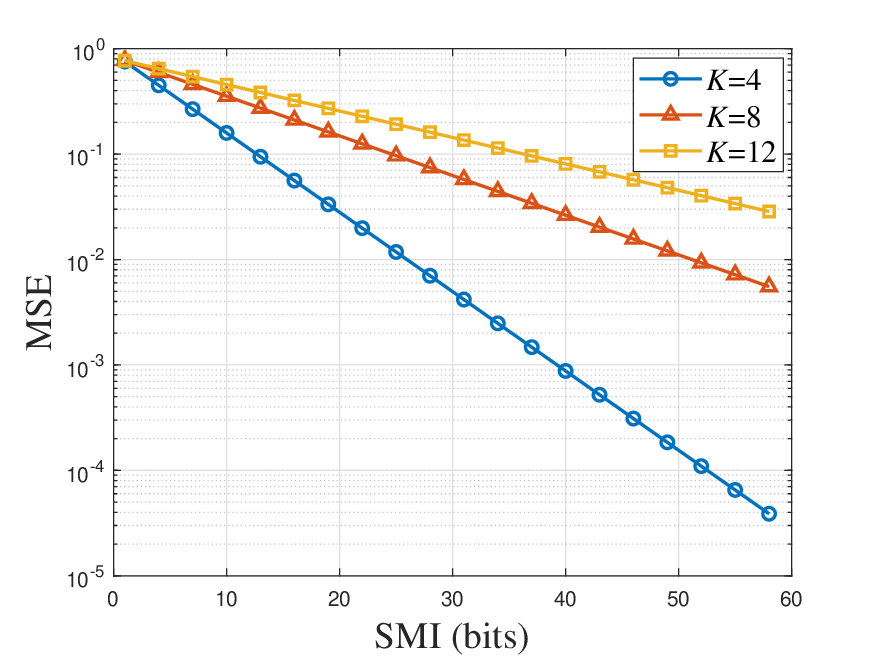}
  \caption{Relationship between SMI and MSE with various values of $K$.}
  \label{SMI_MSE_K}
\end{figure}
\par Fig.~\ref{SMI_MSE_K} presents the impact of the number of sensing parameters to be estimated on the relationship between SMI and MSE, where the cross-correlation coefficient between different sensing parameters is set as $\rho _{\mathrm{s}}=0.3$. It can be seen that the MSE decreases as the SMI increases, which corroborates the intuition that acquiring more information about the sensing parameters can improve the sensing estimation accuracy. Moreover, given SMI, the MSE increases with the number of sensing parameters. This suggests that, for estimating more sensing parameters, it is required to obtain more information about these parameters to maintain the sensing estimation accuracy. For example, when the number of sensing parameters increases from $4$ to $8$, the acquired SMI should increase proportionally from $25$ bit to $50$ bit to ensure a $10^{-2}$-level MSE.

\begin{figure}[htb]
  \centering
  \includegraphics[width=3.3in]{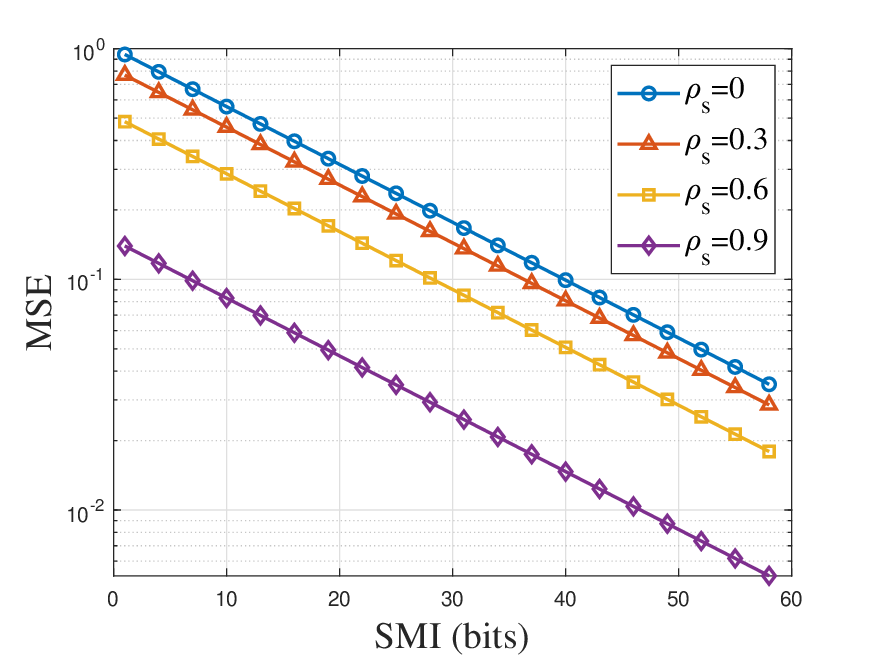}
  \caption{Relationship between SMI and MSE with various values of $\rho _{\mathrm{s}}$.}
  \label{SMI_MSE_rho}
\end{figure}
\par In Fig.~\ref{SMI_MSE_rho}, we fix the number of sensing parameters to be estimated at $K=12$, and study the impact of the parameter correlation on the relationship between SMI and MSE. We can observe that, given SMI, the MSE decreases with the parameter correlation $\rho _{\mathrm{s}}$, especially in the case of high parameter correlation. This is because, in the estimation of multiple sensing parameters, the correlation between parameters introduces redundancy in the acquired information. This redundancy in turn provides diversity gain, which can be exploited to resist the adverse effect of random errors, and ultimately improves the accuracy of sensing parameter estimation.

\subsection{Communication-Sensing Performance Region}
Then, we characterize the communication-sensing performance region, and the following settings are assumed throughout the simulations unless otherwise specified. The ISAC Tx, the sensing Rx, and the communication Rx are each equipped with a uniform linear array, with the numbers of antennas being $N=4$, $M_s=8$, and $M_c=4$, respectively. For each $N$-dimensional sample vector of the transmitted ISAC symbol, its average power is limited to be $P_T=N$ and it follows the complex Gaussian distribution $\mathcal{C} \mathcal{N} \left( 0,P_T\mathbf{I}_N \right) $. The Gaussian channel scenario is considered for both communication and sensing. For the sensing channel vector $\mathbf{h}_{\mathrm{s}}=\mathrm{vec}\left( \mathbf{H}_{\mathrm{s}} \right)$, the first $K=NM_{\mathrm{s}}/2$ elements are the sensing parameters $\mathbf{s}$ to be estimated and the remaining are the parameters $\mathbf{r}$ out of interest. The cross-correlation coefficients of $\mathbf{s}$, of $\mathbf{r}$, and between $\mathbf{s}$ and $\mathbf{r}$ are respectively set as $\rho _{\mathrm{s}}=0.3$, $\rho _{\mathrm{r}}=0.3$, and $\rho _{\mathrm{sr}}=0.2$, respectively. Moreover, the channel gain to noise ratio for communication and sensing are set as $\beta _{\mathrm{c}}=\alpha _{\mathrm{H}_{\mathrm{c}}}^{2}/\sigma _{\mathrm{nc}}^{2}=20$ dB and $\beta _{\mathrm{s}}=\alpha _{\mathrm{H}_{\mathrm{s}}}^{2}/\sigma _{\mathrm{ns}}^{2}=10 $ dB, respectively. In addition, during each CPI, the available REs for the ISAC system are set as $U_{\mathrm{ISAC}}\triangleq \mathrm{card}\left( \mathcal{U} _{\mathrm{ISAC}} \right) =10000$, satisfying $U_{\mathrm{c}}+U_{\mathrm{s}}\triangleq \mathrm{card}\left( \mathcal{U} _{\mathrm{c}} \right) +\mathrm{card}\left( \mathcal{U} _{\mathrm{s}} \right) =U_{\mathrm{ISAC}}$.

\begin{figure}[htbp]
  \centering
  \subfigure[CMI-SMI region.]  
  {
  \label{CMI_SMI_RE}
  \includegraphics[width=3.3in]{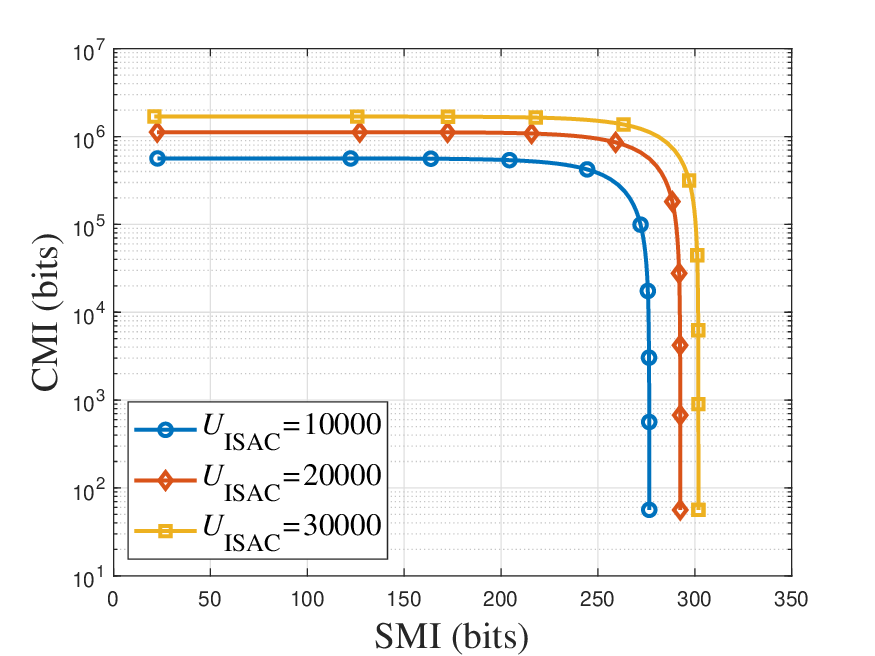}
  }
  \subfigure[CMI-MSE region.] 
  {
  \label{CMI_MSE_RE}
  \includegraphics[width=3.3in]{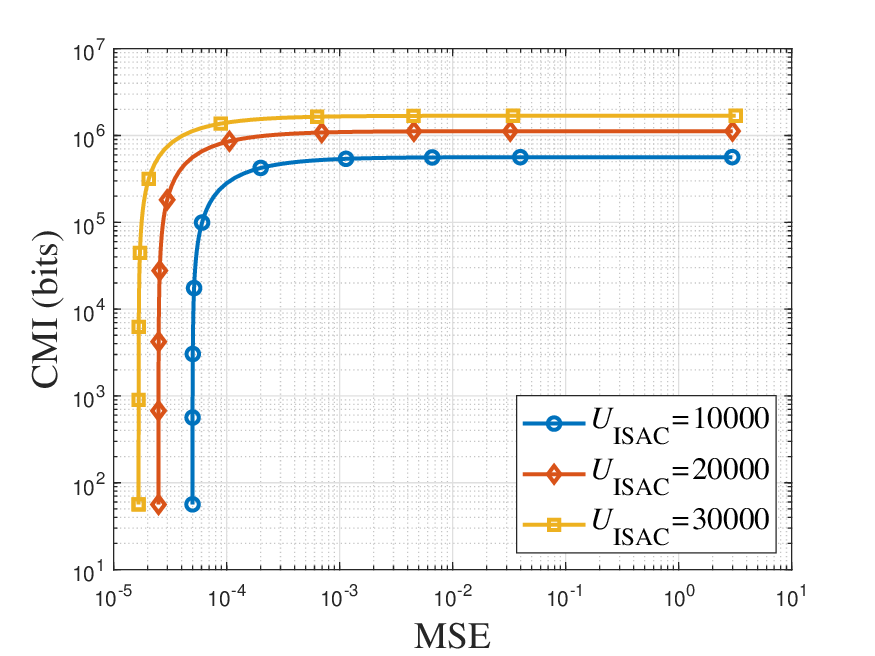}
  }
  \caption{Communication-sensing performance region with various values of $U_{\mathrm{ISAC}}$.}
  \label{C_S_RE}
\end{figure}
Fig.~\ref{CMI_SMI_RE} and Fig.~\ref{CMI_MSE_RE} characterize the CMI-SMI region and CMI-MSE region with different available REs $\mathcal{U} _{\mathrm{ISAC}}$, respectively, where each communication-sensing curve is obtained by changing the allocation of REs between communication and sensing functionalities. One can observe that, given time-frequency resources, there exists a displacement relation between communication performance and sensing performance. This relationship entails that sacrificing the performance of one facilitates enhancement in the performance of another. Moreover, the communication-sensing curve comprises three regions: trade-off region, communication saturation region, and sensing saturation region. In the trade-off region, sacrificing the performance of one can notably boost that of another, while in the communication/sensing saturation region, despite sacrificing the performance of one a lot, little performance gain of another can be obtained. This suggests that, when the time-frequency resources allocated for one is sufficient, additional resources should be allocated to another one for obtaining more performance gain, which corroborates the analysis of time-frequency resource allocation in Section.~\ref{section4}. In addition, with more REs, the communication-sensing performance region becomes larger.

\begin{figure}[htbp]
  \centering
  \subfigure[CMI-SMI region.]  
  {
  \label{CMI_SMI_rhox}
  \includegraphics[width=3.3in]{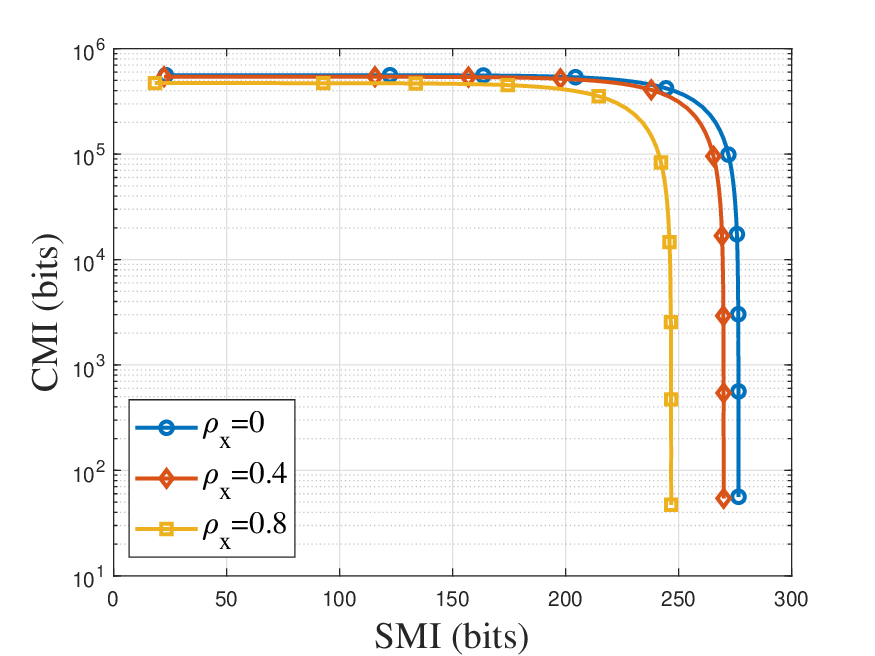}
  }
  \subfigure[CMI-MSE region.] 
  {
  \label{CMI_MSE_rhox}
  \includegraphics[width=3.3in]{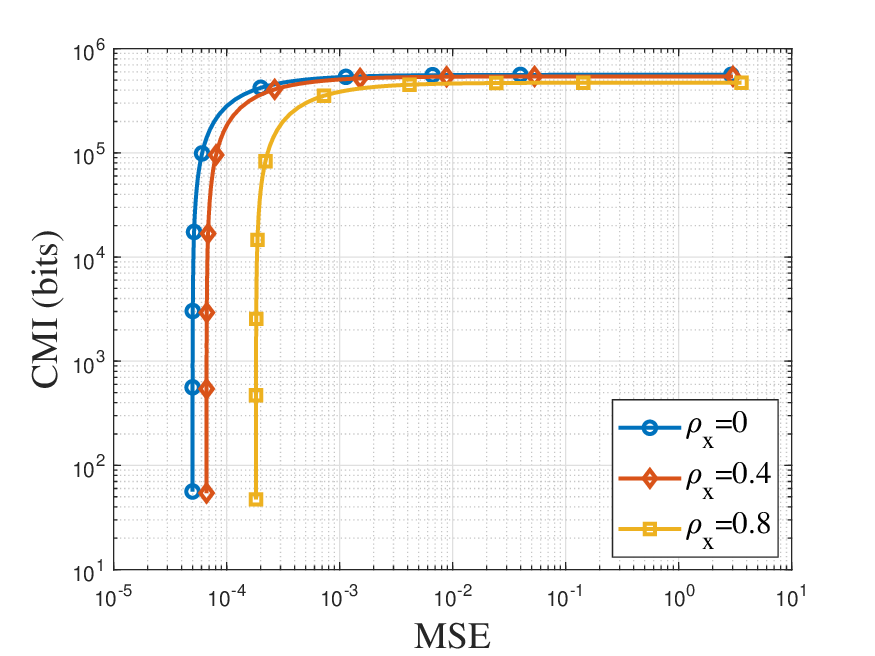}
  }
  \caption{Communication-sensing performance region with various values of $\rho _{\mathrm{x}}$.}
  \label{C_S_rhox}
\end{figure}
\par Then, we consider that the ISAC waveform is spatially correlated, and investigate the impact of the waveform correlation on the communication and sensing performance in Fig.~\ref{C_S_rhox}, where the spatial correlation coefficient of the ISAC waveform is denoted by $\rho _{\mathrm{x}}$. We see that decreasing the spatial correlation coefficient $\rho _{\mathrm{x}}$ is beneficial to both communication and sensing. This demonstrates the unified side existed in the design of the communication waveform and the sensing waveform as revealed in Proposition~\ref{proposition2}, namely, an ISAC waveform with lower correlation can not only carry more information to improve the communication efficiency, but can also acquire more independent measurements of sensing parameters to enhance the parameter estimation task.

\begin{figure}[htbp]
  \centering
  \subfigure[Communication performance.]  
  {
  \label{waveform_c}
  \includegraphics[width=3.3in]{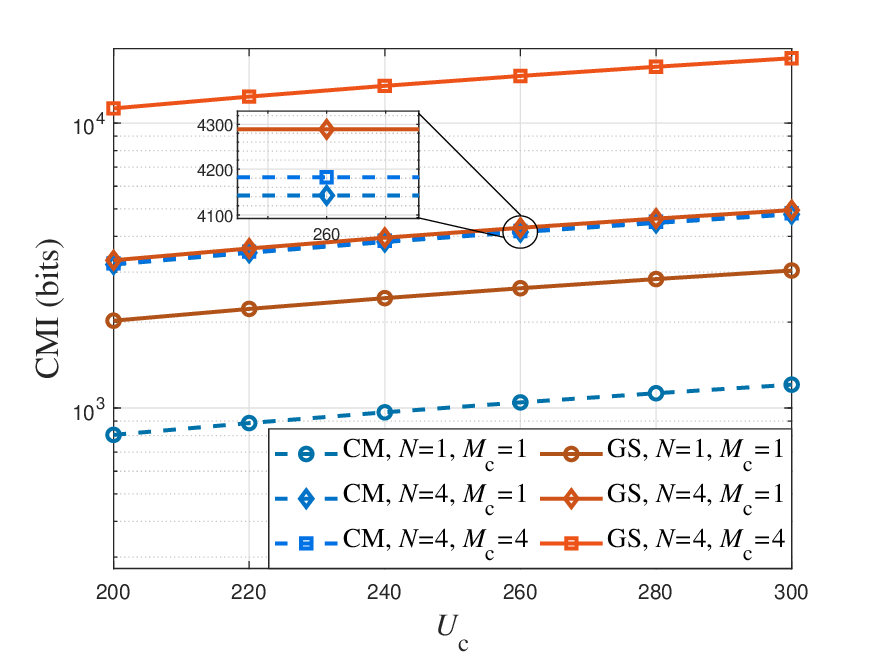}
  }
  \subfigure[Sensing performance.] 
  {
  \label{waveform_s}
  \includegraphics[width=3.3in]{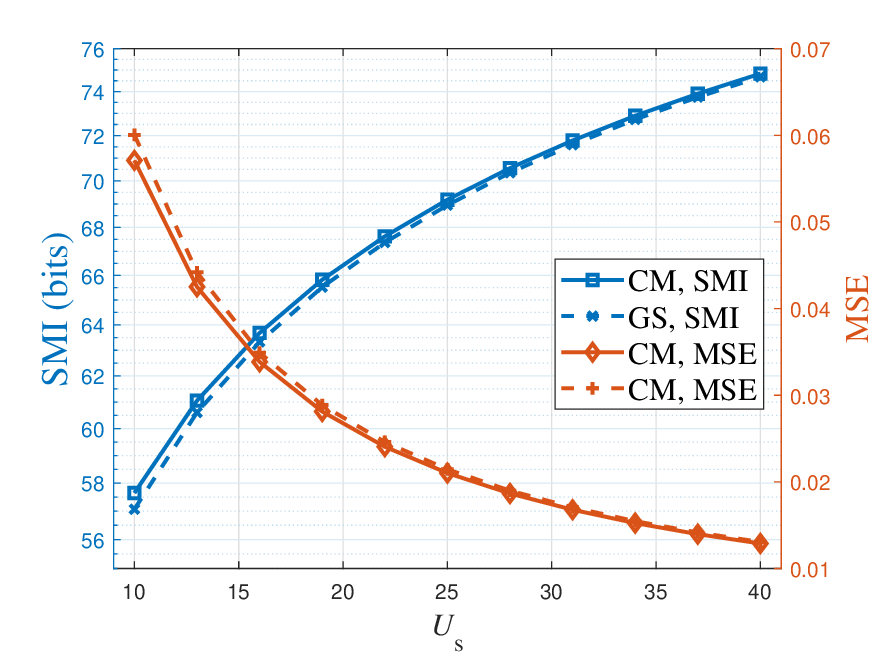}
  }
  \caption{Performance comparison of the constant-modulus waveform and Gaussian waveform.}
  \label{waveform}
\end{figure}
Next, in Fig.~\ref{waveform}, we study the impact of the probability distribution function of the waveform amplitude on communication performance and sensing performance, respectively. We compare the performance of the Gaussian waveform (denoted by GS) and the constant-modulus waveform (denoted by CM) with uniformly distributed phases. As shown in Fig.~\ref{waveform_c}, in the single-antenna communication Rx scenario (i.e., $M_{\mathrm{c}}=1$), the Gaussian waveform has better communication performance than the constant-modulus waveform when $N=1$, which originates from the fact that the Gaussian distribution is the maximum entropy distribution. However, a noteworthy phenomenon is that the communication performance gap between the two waveforms gradually vanishes as the number of antennas at the ISAC Tx $N$ increases. This is because the communication Rx with a single antenna can not distinguish the multiple streams from the ISAC Tx, and according to the central limit theorem, the sum of multiple constant-modulus signals gradually approximately follows a complex Gaussian distribution as $N$ increases. While in the multi-antenna communication Rx scenario (i.e., $M_{\mathrm{c}}=4$), the communication Rx is able to distinguish the multiple streams from the ISAC Tx, and the communication benefits of the Gaussian waveform reemerge. Moreover, for sensing performance, as shown in Fig.~\ref{waveform_s}, the Gaussian waveform performs worse than the constant-modulus waveform when fewer REs are allocated to the sensing functionality, due to the adverse effect of signal randomness as analyzed in Proposition~\ref{proposition2}. As the allocated REs increase, the sensing Rx collects more data for parameter estimation, thereby mitigating the adverse effect caused by the randomness of the waveform amplitude, which consequently diminishes the sensing performance gap between the two waveforms. These phenomena prove the contradiction in the waveform design of communication and sensing revealed in Proposition~\ref{proposition2}, namely, the communication functionality prefers a more random amplitude to carry more information, while the sensing functionality prefers a constant waveform amplitude to ensure a stable estimation of sensing parameters.

\begin{figure}[htb]
  \centering
  \includegraphics[width=3.3in]{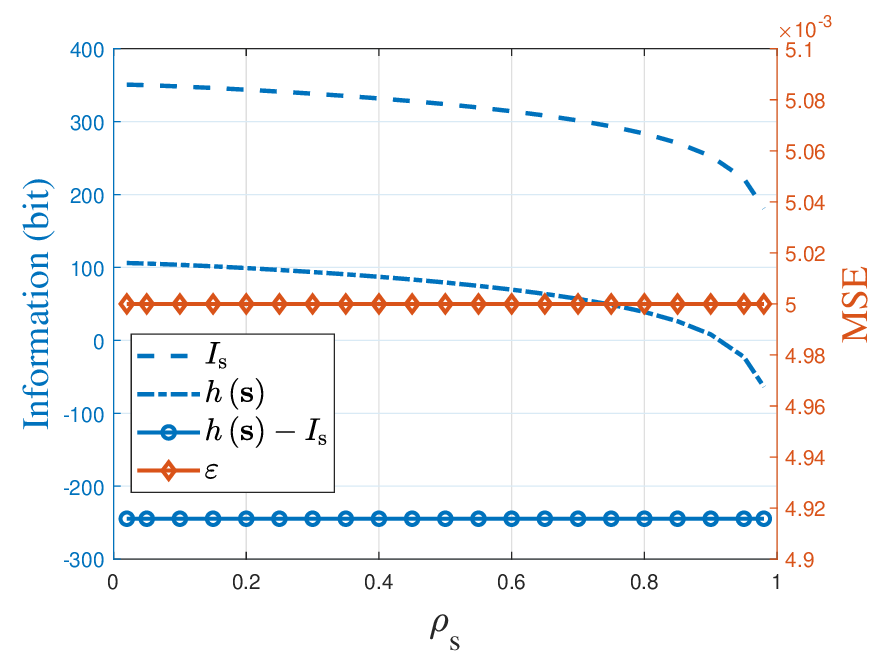}
  \caption{Sensing performance with various values of $\rho _{\mathrm{s}}$.}
  \label{Sensing_rhos}
\end{figure}

Finally, in Fig.~\ref{Sensing_rhos}, we investigate the impact of the sensing parameter correlation on the performance of sensing channel estimation based on Proposition~\ref{proposition1}, where we set $K=NM_{\mathrm{s}}$ and $\mathrm{card}\left( \mathcal{U} _{\mathrm{s}} \right) =100$. For ease of notation, we denote  $h\left( \mathbf{s} \right) =\log\det\mathrm{(}\mathbf{R}_{\mathrm{s}})$ as the entropy of the sensing parameters. We can observe that the SMI decreases with the correlation of the sensing parameters, because more correlated parameters have more information redundancy. However, as spatial correlation increases, the reduction in the entropy of the sensing parameters is the same as that in the SMI, leaving the MSE nearly unchanged. An intuitive interpretation is that, when the correlation between sensing parameters is stronger, although less SMI can be acquired, the signal processing gain of joint parameter estimation is enhanced, which ultimately stabilizes the MSE.

\section{Conclusion}\label{section6}
In this paper, we established an information model for band-limited discrete-time ISAC systems, developed communication-sensing regions for fundamental limits investigation, demonstrated the paradoxical balance in the ISAC waveform design, as well as revealed the impact of the time-frequency-spatial resource allocation on the performance trade-off. In particular, we first leveraged the information theory alongside the Nyquist sampling theorem to establish a unified information model for the band-limited continuous-time ISAC system, which incorporates both temporal, spectral, and spatial properties. Under this framework, we derived the SMI for sensing performance characterization and revealed its connection with the MSE. Then, we proposed the CMI-SMI and CMI-MSE regions to investigate the performance boundary of ISAC and the trade-off between communication and sensing. Through theoretical analysis and numerical results, two valuable insights were revealed for guiding the design of practical ISAC systems. First, for ISAC waveform design, the communication functionality prefers a random amplitude to convey more information, while the sensing functionality prefers a constant-modulus waveform to ensure a stable parameter estimation. In contrast, both functionalities prefer a low-correlation waveform with random phases, which not only improves communication efficiency but also provides more independent measurements of sensing parameters. Second, for time-frequency resource allocation, there exists a linear positive proportional relationship between the allocated time-frequency resource and the achieved communication rate/sensing MSE. Moreover, in sensing channel estimation, the MSE is not affected by the sensing parameter correlation, since this correlation provides signal processing gain in joint parameter estimation despite reducing the acquired SMI.

\bibliographystyle{IEEEtran}
\bibliography{references}{}

\begin{thebibliography}{10}
\providecommand{\url}[1]{#1}
\csname url@samestyle\endcsname
\providecommand{\newblock}{\relax}
\providecommand{\bibinfo}[2]{#2}
\providecommand{\BIBentrySTDinterwordspacing}{\spaceskip=0pt\relax}
\providecommand{\BIBentryALTinterwordstretchfactor}{4}
\providecommand{\BIBentryALTinterwordspacing}{\spaceskip=\fontdimen2\font plus
\BIBentryALTinterwordstretchfactor\fontdimen3\font minus \fontdimen4\font\relax}
\providecommand{\BIBforeignlanguage}[2]{{%
\expandafter\ifx\csname l@#1\endcsname\relax
\typeout{** WARNING: IEEEtran.bst: No hyphenation pattern has been}%
\typeout{** loaded for the language `#1'. Using the pattern for}%
\typeout{** the default language instead.}%
\else
\language=\csname l@#1\endcsname
\fi
#2}}
\providecommand{\BIBdecl}{\relax}
\BIBdecl

\bibitem{LiuIntegrated2022}
F.~Liu, Y.~Cui, C.~Masouros, J.~Xu, T.~X. Han, Y.~C. Eldar, and S.~Buzzi, ``Integrated sensing and communications: Toward dual-functional wireless networks for {6G} and beyond,'' \emph{IEEE Journal on Selected Areas in Communications}, vol.~40, no.~6, pp. 1728--1767, Jun. 2022.

\bibitem{ZhangOverview}
J.~A. Zhang, F.~Liu, C.~Masouros, R.~W. Heath, Z.~Feng, L.~Zheng, and A.~Petropulu, ``An overview of signal processing techniques for joint communication and radar sensing,'' \emph{IEEE Journal of Selected Topics in Signal Processing}, vol.~15, no.~6, pp. 1295--1315, Nov. 2021.

\bibitem{LiuSeventy2023}
F.~Liu, L.~Zheng, Y.~Cui, C.~Masouros, A.~P. Petropulu, H.~Griffiths, and Y.~C. Eldar, ``Seventy years of radar and communications: The road from separation to integration,'' \emph{IEEE Signal Processing Magazine}, vol.~40, no.~5, pp. 106--121, Jul. 2023.

\bibitem{ChiriyathRadar}
A.~R. Chiriyath, B.~Paul, and D.~W. Bliss, ``Radar-communications convergence: Coexistence, cooperation, and co-design,'' \emph{IEEE Transactions on Cognitive Communications and Networking}, vol.~3, no.~1, pp. 1--12, Mar. 2017.

\bibitem{FengJoint2020}
Z.~Feng, Z.~Fang, Z.~Wei, X.~Chen, Z.~Quan, and D.~Ji, ``Joint radar and communication: A survey,'' \emph{China Communications}, vol.~17, no.~1, pp. 1--27, Jan. 2020.

\bibitem{LuIntegrated}
S.~Lu, F.~Liu, Y.~Li, K.~Zhang, H.~Huang, J.~Zou, X.~Li, Y.~Dong, F.~Dong, J.~Zhu, Y.~Xiong, W.~Yuan, Y.~Cui, and L.~Hanzo, ``Integrated sensing and communications: Recent advances and ten open challenges,'' \emph{IEEE Internet of Things Journal}, vol.~11, no.~11, pp. 19\,094--19\,120, Jun. 2024.

\bibitem{LiuSurvey}
A.~Liu, Z.~Huang, M.~Li, Y.~Wan, W.~Li, T.~X. Han, C.~Liu, R.~Du, D.~K.~P. Tan, J.~Lu, Y.~Shen, F.~Colone, and K.~Chetty, ``A survey on fundamental limits of integrated sensing and communication,'' \emph{IEEE Communications Surveys \& Tutorials}, vol.~24, no.~2, pp. 994--1034, Feb. 2022.

\bibitem{OuyangIntegrated}
C.~Ouyang, Y.~Liu, H.~Yang, and N.~Al-Dhahir, ``Integrated sensing and communications: A mutual information-based framework,'' \emph{IEEE Communications Magazine}, vol.~61, no.~5, pp. 26--32, May 2023.

\bibitem{ChalisePerformance2017}
B.~K. Chalise, M.~G. Amin, and B.~Himed, ``Performance tradeoff in a unified passive radar and communications system,'' \emph{IEEE Signal Processing Letters}, vol.~24, no.~9, pp. 1275--1279, Sep. 2017.

\bibitem{ChalisePerformance2018}
B.~K. Chalise and B.~Himed, ``Performance tradeoff in a unified multi-static passive radar and communication system,'' in \emph{Proc. 2018 IEEE Radar Conference (RadarConf18)}, 2018, pp. 0653--0658.

\bibitem{XiaoWaveform2022}
Z.~Xiao and Y.~Zeng, ``Waveform design and performance analysis for full-duplex integrated sensing and communication,'' \emph{IEEE Journal on Selected Areas in Communications}, vol.~40, no.~6, pp. 1823--1837, Jun. 2022.

\bibitem{XiongFlowing}
Y.~Xiong, F.~Liu, Y.~Cui, W.~Yuan, and T.~X. Han, ``Flowing the information from shannon to fisher: Towards the fundamental tradeoff in {ISAC},'' in \emph{Proc. 2022 IEEE Global Communications Conference}, 2022, pp. 5601--5606.

\bibitem{HuaGLOBECOM}
H.~Hua, X.~Song, Y.~Fang, T.~X. Han, and J.~Xu, ``{MIMO} integrated sensing and communication with extended target: {CRB-Rate} tradeoff,'' in \emph{Proc. 2022 IEEE Global Communications Conference}, Apr. 2022, pp. 4075--4080.

\bibitem{RenFundamental}
Z.~Ren, Y.~Peng, X.~Song, Y.~Fang, L.~Qiu, L.~Liu, D.~W.~K. Ng, and J.~Xu, ``Fundamental {CRB}-rate tradeoff in multi-antenna {ISAC} systems with information multicasting and multi-target sensing,'' \emph{IEEE Transactions on Wireless Communications}, vol.~23, no.~4, pp. 3870--3885, Apr. 2024.

\bibitem{LiuCram}
F.~Liu, Y.-F. Liu, A.~Li, C.~Masouros, and Y.~C. Eldar, ``{Cramér-Rao} bound optimization for joint radar-communication beamforming,'' \emph{IEEE Transactions on Signal Processing}, vol.~70, pp. 240--253, Jun. 2022.

\bibitem{HuaMIMO}
H.~Hua, T.~X. Han, and J.~Xu, ``{MIMO} integrated sensing and communication: {CRB-Rate} tradeoff,'' \emph{IEEE Transactions on Wireless Communications}, vol.~23, no.~4, pp. 2839--2854, Apr. 2024.

\bibitem{AnFundamental2023}
J.~An, H.~Li, D.~W.~K. Ng, and C.~Yuen, ``Fundamental detection probability vs. achievable rate tradeoff in integrated sensing and communication systems,'' \emph{IEEE Transactions on Wireless Communications}, vol.~22, no.~12, pp. 9835--9853, Dec. 2023.

\bibitem{DongningMutual}
D.~Guo, S.~Shamai, and S.~Verdu, ``Mutual information and minimum mean-square error in gaussian channels,'' \emph{IEEE Transactions on Information Theory}, vol.~51, no.~4, pp. 1261--1282, Apr. 2005.

\bibitem{TangMIMO2010}
B.~Tang, J.~Tang, and Y.~Peng, ``{MIMO} radar waveform design in colored noise based on information theory,'' \emph{IEEE Transactions on Signal Processing}, vol.~58, no.~9, pp. 4684--4697, Sep. 2010.

\bibitem{AhmadipourInformation}
M.~Ahmadipour, M.~Kobayashi, M.~Wigger, and G.~Caire, ``An information-theoretic approach to joint sensing and communication,'' \emph{IEEE Transactions on Information Theory}, vol.~70, no.~2, pp. 1124--1146, Feb. 2024.

\bibitem{AhmadipourJoint}
M.~Ahmadipour, M.~Wigger, and M.~Kobayashi, ``Joint sensing and communication over memoryless broadcast channels,'' in \emph{Proc. 2020 IEEE Information Theory Workshop (ITW)}, 2021, pp. 1--5.

\bibitem{AhmadipourSystems}
M.~Ahmadipour and M.~Wigger, ``An information-theoretic approach to collaborative integrated sensing and communication for two-transmitter systems,'' \emph{IEEE Journal on Selected Areas in Information Theory}, vol.~4, pp. 112--127, Jun. 2023.

\bibitem{GuerciJoint}
J.~R. Guerci, R.~M. Guerci, A.~Lackpour, and D.~Moskowitz, ``Joint design and operation of shared spectrum access for radar and communications,'' in \emph{Proc. 2015 IEEE Radar Conference (RadarCon)}, 2015, pp. 0761--0766.

\bibitem{ChiriyathInner}
A.~R. Chiriyath, B.~Paul, G.~M. Jacyna, and D.~W. Bliss, ``Inner bounds on performance of radar and communications co-existence,'' \emph{IEEE Transactions on Signal Processing}, vol.~64, no.~2, pp. 464--474, Jan. 2016.

\bibitem{BlissCooperative}
D.~W. Bliss, ``Cooperative radar and communications signaling: The estimation and information theory odd couple,'' in \emph{Proc. 2014 IEEE Radar Conference}, 2014, pp. 0050--0055.

\bibitem{YangMIMOradar}
Y.~Yang and R.~S. Blum, ``{MIMO} radar waveform design based on mutual information and minimum mean-square error estimation,'' \emph{IEEE Transactions on Aerospace and Electronic Systems}, vol.~43, no.~1, pp. 330--343, Jan. 2007.

\bibitem{LiuDeterministic-Random}
F.~Liu, Y.~Xiong, K.~Wan, T.~X. Han, and G.~Caire, ``Deterministic-random tradeoff of integrated sensing and communications in gaussian channels: A rate-distortion perspective,'' in \emph{Proc. 2023 IEEE International Symposium on Information Theory (ISIT)}, 2023, pp. 2326--2331.

\bibitem{KumariPerformance}
P.~Kumari, D.~H.~N. Nguyen, and R.~W. Heath, ``Performance trade-off in an adaptive {IEEE} 802.11{AD} waveform design for a joint automotive radar and communication system,'' in \emph{Proc. 2017 IEEE International Conference on Acoustics, Speech and Signal Processing (ICASSP)}, 2017, pp. 4281--4285.

\bibitem{YuNon-Orthogonal}
Z.~Yu, X.~Hu, C.~Liu, and M.~Peng, ``{IRS}-aided non-orthogonal {ISAC} systems: Performance analysis and beamforming design,'' \emph{IEEE Transactions on Green Communications and Networking}, pp. 1--1, May 2024.

\bibitem{tse2005fundamentals}
D.~Tse and P.~Viswanath, \emph{Fundamentals of wireless communication}.\hskip 1em plus 0.5em minus 0.4em\relax Cambridge university press, 2005.

\bibitem{cover1999elements}
T.~M. Cover, \emph{Elements of information theory}.\hskip 1em plus 0.5em minus 0.4em\relax John Wiley \& Sons, 1999.

\end{thebibliography}
\end{document}